\documentclass[onefignum,onetabnum]{siamsials251208}
\usepackage{tikz}
\usepackage{enumitem}
\usepackage{algorithm}
\usepackage{algpseudocode}
\usepackage{amsmath}
\usepackage{multirow}

\usepackage{amsfonts}
\usepackage{graphics}
\usepackage{xcolor} % For cell coloring
\usepackage{amssymb}
\usepackage{psfrag}
\usepackage{multicol}
\usepackage{multirow}
\usepackage{graphicx}
\usepackage[utf8]{inputenc}
\usepackage{bm}
\usepackage{mathtools}
\usepackage{epstopdf}
\usepackage{setspace}
\usepackage{amsmath,amssymb,bbm}
\usepackage[utf8]{inputenc}
\usepackage{url}
\usepackage{color}
\usepackage{booktabs}

\DeclareMathOperator{\Ima}{Im}

\newsiamremark{remark}{Remark}
\newsiamremark{hypothesis}{Hypothesis}
\crefname{hypothesis}{Hypothesis}{Hypotheses}
\newsiamthm{claim}{Claim}
\newsiamthm{assum}{Assumption}
\newsiamthm{scheme}{Scheme}
\newsiamthm{lem}{Lemma}
\newsiamthm{thm}{Theorem}
\headers{PHL: Persistent Hyperdigraph Learning}{X.~Xu, C.~Wang, J.~Chen }

\title{PHL: Persistent Hyperdigraph Learning for Protein–Protein Binding Affinity Prediction}

\author{
Xingjian Xu\thanks{
Department of Mathematics, University of Florida, Gainesville, FL 32611, USA 
({\tt xingjianxu@ufl.edu}).}
\and 
Chunmei Wang%\footnotemark[3]
\thanks{
Department of Mathematics, University of Florida, Gainesvile, FL 32611, USA ({\tt chunmei.wang@ufl.edu)}.}
 \and 
Jiahui Chen\thanks{Department of Mathematical Sciences, University of Arkansas, Fayetteville, AR 72701, USA ({\tt jiahuic@uark.edu}).}
}

\usepackage{bbm}

\allowdisplaybreaks

\begin{document}

\maketitle

\begin{abstract}
Persistent homology and persistent Laplacians have become effective descriptors of large molecular binding interactions, the latter encoding multiscale geometry and spectral information beyond the harmonic subspace. Both, however, are built on undirected graphs and pairwise contacts, while the interactions that determine binding are often directional and involve more than two atoms at once. We introduce persistent hyperdigraph learning (PHL), which represents these directed, many-body interactions directly. Because hyperdigraph Laplacian spectra become costly to compute at higher topological orders, we further use stochastic trace estimation to develop a matrix-free formulation (MFPHL) that estimates Laplacian trace statistics from probe-based quadratic forms evaluated through sparse boundary operators, without assembling or diagonalizing. The number of probes required for a target relative accuracy of the trace estimate is independent of matrix dimension, and per-probe cost scales with operator sparsity. On protein–protein binding affinity benchmarks, MFPHL reduces feature-generation time by roughly two orders of magnitude relative to the eigenvalue-based pipeline. Its predictive accuracy matches that of spectral PHL on the P2P wild-type set within training-seed variability but is lower on the two larger datasets. This dataset-dependent trade-off brings higher-order hyperdigraph features within practical reach.
\end{abstract}

\begin{keywords} Persistent Laplacians; Protein-protein binding; Persistent hyperdigraph learning; Binding affinity prediction.
\end{keywords}

\section{Introduction}
The formation of specific protein-protein complexes governs signal transduction \cite{good2011scaffold,pawson1995protein}, immune recognition \cite{rudolph2006how,stone2009tcell,
sundberg2002molecular,vandermerwe2003molecular}, transcriptional regulation \cite{jolma2015dna,lambert2018human,ptashne1997transcriptional,spitz2012transcription}, and the assembly of macromolecular machines \cite{ahnert2015principles,alberts1998cell,marsh2015structure,robinson2007molecular}. Each such complex is characterized by its binding affinity, given by the equilibrium dissociation constant $K_d$ or the associated standard free energy $\Delta G^{\circ} = RT \ln (K_d/c^{\circ})$, where $R$ is the gas constant, $T$ is the absolute temperature, and $c^{\circ}=1$\,M is the standard concentration. Since affinity determines the occupancy of a complex at physiological concentrations, it is the quantity that connects interface structure to biological function and the structure of a complex carries much of the information needed to predict it. Affinity also depends on solution conditions such as temperature, pH, and ionic strength, and reflects a conformational ensemble rather than a single structure. Quantitative prediction supports the optimization of therapeutic antibodies and engineered binders~ \cite{cao2022design,kuroda2012computer} and allows interface mutations to be interpreted on a continuous scale rather than classified only as tolerated or disruptive~\cite{starr2020deep}. However, direct measurement by calorimetry or surface plasmon resonance requires purified protein and substantial instrument time per complex \cite{kastritis2013binding,schuck1997use,velazquez2006isothermal}, which places large-scale surveys of complexes and mutations beyond experimental reach.

Computational alternatives range from physics-based free energy
calculations, which are well founded but costly and sensitive to
force-field and sampling choices
\cite{cournia2017relative,genheden2015mmpbsa}, to empirical scoring
functions such as FoldX and Rosetta, which are fast but generalize
unevenly across interface types
\cite{alford2017rosetta,schymkowitz2005foldx}. Data-driven models have since become dominant, progressing from hand-designed descriptors
\cite{moal2011protein} to representations learned from molecular
surfaces \cite{gainza2020deciphering} or protein language models
\cite{lin2023evolutionary}. What such models can learn is bounded by how structure is encoded, and algebraic topology has proved an effective encoding: persistent homology tracks connected components, loops, and cavities across a filtration of nested simplicial complexes, yielding a stable, coordinate-free description of shape
\cite{edelsbrunner2002topological,zomorodian2005computing}, and
element-specific variants that label atoms by chemical type have achieved leading performance in blind affinity prediction
\cite{cang2018integration,nguyen2019mathematical}. Such descriptions are purely harmonic. Persistent Laplacians retain the harmonic information in the kernel, while the remaining eigenvalues supply geometric detail beyond topology \cite{memoli2022persistent,wang2020persistent}, and the construction extends to hypergraphs
\cite{bressan2019embedded}.

These representations nonetheless build their cells from pairwise proximity. Contact maps, distance descriptors, and the graphs and simplicial filtrations underlying both graph neural networks and topological descriptors are generated by thresholding interatomic distances, so that even when atoms carry chemical labels, the groups themselves are never primitive objects of the representation. Binding, by contrast, is cooperative \cite{horovitz1996double}: hydrophobic clusters \cite{chandler2005interfaces}, aromatic stacking \cite{burley1985aromatic}, and salt bridges formed between carboxylate and amine groups \cite{gvritishvili2008cooperativity} derive their energetic contribution from the collective arrangement of several atoms at once, and the many-body structure of this kind enters only implicitly. Interactions across an interface are moreover asymmetric: nitrogen and oxygen play distinguishable chemical roles in a polar contact \cite{mcdonald1994satisfying}, and the two partners of a complex are not interchangeable \cite{chakrabarti2002dissecting}. Hyperdigraphs make both features explicit: a hyperedge relates an arbitrary number of vertices simultaneously, so a chemical group can be represented as a single object, while the directed formulation retains the asymmetry that undirected representations discard. Whether this additional expressiveness translates into predictive gain is an empirical question, and one that has not been tested at high topological order.

Building on the formulation of persistent hyperdigraph homology and its associated Laplacians \cite{chen2023persistent}, we introduce persistent hyperdigraph learning (PHL) for protein–protein binding affinity prediction. PHL represents molecular interfaces through directed hyperedges that group atoms by elemental composition, with orientations determined by electronegativity. In the representation tested here, the hyperedges are ordered cliques of an interatomic distance graph within each element-pair channel, not independently identified chemical groups. The experiments, therefore, assess this clique-derived representation rather than the full expressive capacity of general hyperdigraphs. This construction incorporates directed, many-body relationships into a multiscale representation, whose Laplacian spectra capture both harmonic and non-harmonic information. On the P2P benchmark \cite{xu2025plnet}, PHL achieves a higher correlation than the published results of previously reported methods. However, the rapid growth of hyperedge counts with topological order makes repeated Laplacian assembly and eigendecomposition expensive: the spectral pipeline takes about $3.6\times10^{5}$\,s on the 343 WT complexes, and in the largest channels, the operators must be truncated before diagonalization in Section~\ref{results}.

To overcome this computational limitation, we develop matrix-free persistent hyperdigraph learning (MFPHL), which replaces eigendecomposition with stochastic estimation of the first two spectral moments \cite{hutchinson1989stochastic}. By exploiting the Laplacian's factorization through sparse boundary operators, MFPHL computes these estimates without explicitly assembling the Laplacian, retaining non-harmonic spectral information at a cost governed by operator sparsity. The required probe count admits dimension-independent bounds for prescribed relative accuracy and confidence of the trace estimate for each operator. MFPHL reduces feature-generation time by approximately two orders of magnitude. Its predictive performance is within the training-seed variability of the spectral PHL baseline on WT but lower on V2020 and P2P$_{\mathrm{WT}}$. It makes higher-order descriptors substantially cheaper to compute on all three datasets.

The remainder of this article is structured as follows. Section~\ref{Methods} presents the mathematical foundations of PHL and MFPHL, including hyperdigraph homology, persistent Laplacians, and stochastic spectral moment estimation. Section~\ref{MFPHL_framework} describes the construction of PHL representations and their matrix-free implementation for binding affinity prediction. Section~\ref{results} evaluates predictive accuracy and compares the feature-generation costs of PHL and MFPHL. Section~\ref{discussion} examines what the descriptors encode, why Betti numbers are omitted, and the practical limits of increasing topological order. Section~\ref{Conclusion} summarizes the findings and outlines future directions.

\section{Modeling}\label{Methods}
This section develops the proposed method in three steps. We first recall hyperdigraphs and their Laplacians, which encode many-body and directed structures in a single algebraic object. We then introduce a filtration and the associated persistent Laplacians, giving a multiscale description of an interface. Finally, we develop the stochastic trace estimator that extracts spectral descriptors, together with an error bound, from the boundary operators alone and never forms the Laplacian. The first two steps determine what is described; the third determines how it is computed and at what cost.

\subsection{Hyperdigraphs and hyperdigraph Laplacians}\label{hyperdigraph}

We recall the hyperdigraph construction of \cite{chen2023persistent}, which uses directed hyperedges to encode asymmetric relations among vertices. Throughout, let $V$ be a finite, nonempty ordered vertex set;
let $\mathbf{P}_{n}(V)$ denote the collection of $n$-element subsets of $V$; let $\Sigma_{n}$ denote the symmetric group on $n$ letters; and let $\mathbb{K}$ be a field.

\begin{definition}[Directed hyperedge~\cite{chen2023persistent}] \label{def_hyperedge}
A directed $p$-hyperedge on $V$ is a pair \((\sigma,e)\in\Sigma_{p+1}\times\mathbf{P}_{p+1}(V),\) where $e$ is a set of $p+1$ vertices and $\sigma$ specifies an ordering of those vertices. Since $V$ is ordered, the subset $e$ inherits an
ordering, and applying $\sigma$ to this induced ordering determines a sequence \(v_{0}v_{1}\cdots v_{p}\)
of distinct vertices. We write
\(
\mathbf{S}(V)
=
\coprod_{n=1}^{|V|}
\Sigma_{n}\times\mathbf{P}_{n}(V)
\) for the collection of all directed hyperedges on $V$.
\end{definition}

\begin{definition}[Hyperdigraph~\cite{chen2023persistent}]
\label{def_hyperdigraph}
A hyperdigraph on $V$ is a pair \(\vec{\mathcal{H}}=(V,\vec{E}),\) where $\vec{E}\subseteq\mathbf{S}(V)$. We denote by
$\vec{E}^{\,p}$ the collection of directed $p$-hyperedges of $\vec{\mathcal{H}}$. If every permutation occurring in $\vec{E}$ is the
identity permutation, then $\vec{\mathcal{H}}$ reduces to a hypergraph
in the usual sense. A {morphism}
\(
f:\vec{\mathcal{H}}\longrightarrow\vec{\mathcal{H}}'
\) is a map $f:V\to V'$ such that
\(
(\sigma,f(e))\in\vec{E}'
\quad
\text{for every }(\sigma,e)\in\vec{E}.
\) Here $f(e)=\{f(v):v\in e\}$ is regarded as a subset of the ordered vertex set $V'$. Consequently, $f$ is injective on each hyperedge,
because $(\sigma,f(e))$ must contain the same number of vertices as $(\sigma,e)$.
\end{definition}

\begin{remark}
The morphism convention above follows \cite{chen2023persistent}: the permutation $\sigma$ is preserved while $f$ is applied to the underlying vertex set $e$. In the filtrations considered below, all morphisms are inclusions arising from enlargement of the directed hyperedge set, for which this convention agrees with the corresponding sequence-based description.
\end{remark}

Since the vertices of a directed hyperedge are distinct, no regularity condition is required. To obtain a boundary operator, the hyperedges are first placed inside a chain complex generated by all such sequences.
 
\begin{definition}[Sequence chain complex~\cite{chen2023persistent}]\label{def_boundary}
Let $S_{p}(V;\mathbb{K})$ denote the $\mathbb{K}$-linear space generated by $\Sigma_{p+1}\times\mathbf{P}_{p+1}(V)$. The boundary operator $d_{p} : S_{p}(V;\mathbb{K}) \to S_{p-1}(V;\mathbb{K})$ is
\begin{equation}\label{eq_boundary}
  d_{p} = \sum_{i=0}^{p}(-1)^{i}\partial_{i},
  \qquad
  \partial_{i}(x_{0},x_{1},\dots,x_{p})
  = (x_{0},\dots,\widehat{x_{i}},\dots,x_{p}),
\end{equation}
where $\widehat{x_{i}}$ indicates omission of $x_{i}$, with the convention $d_{0}e = 0$ for $e \in V$. Since $d_{p-1}\circ d_{p} = 0$, the family $S_{\ast}(V;\mathbb{K}) = (S_{p}(V;\mathbb{K}))_{p\geq 0}$ is a chain complex.
\end{definition}
 
Let \(F_{p}(\vec{\mathcal{H}};\mathbb{K})\) denote the subspace of \(S_{p}(V;\mathbb{K})\) generated by \(\vec{E}^{\,p}\), and \(F_{\ast}(\vec{\mathcal{H}};\mathbb{K})\) is the resulting graded space. The operator \(d_{p}\) restricts to a map
\(F_{p}(\vec{\mathcal{H}};\mathbb{K}) \to S_{p-1}(V;\mathbb{K})\), but its image need not lie in \(F_{p-1}(\vec{\mathcal{H}};\mathbb{K})\).
 
\begin{remark}\label{rem_no_closure}
A simplicial complex is closed under the passage to faces, which is what makes the boundary operator well behaved on it. A hyperdigraph carries no such closure, and this is precisely what allows a chemical group to be a single primitive object of the representation. The cost is that \(d_{p}\big(F_{p}(\vec{\mathcal{H}};\mathbb{K})\big) \not\subseteq
F_{p-1}(\vec{\mathcal{H}};\mathbb{K})\) in general, so \(F_{\ast}(\vec{\mathcal{H}};\mathbb{K})\) together with \(d\) is not a chain complex and its homology is undefined.
\end{remark}
 
Following the embedded homology of hypergraphs \cite{bressan2019embedded}, one retains the largest subspace on which \(d\) is well defined.
 
\begin{definition}[Infimum complex]\label{def_omega}
The infimum chain group of \(\vec{\mathcal{H}}\) in degree \(p\) is
\begin{equation}\label{eq_omega}
  \Omega_{p}(\vec{\mathcal{H}};\mathbb{K})
  = F_{p}(\vec{\mathcal{H}};\mathbb{K})
    \cap d^{-1}F_{p-1}(\vec{\mathcal{H}};\mathbb{K})
  = \big\{\, x \in F_{p}(\vec{\mathcal{H}};\mathbb{K})
    \;:\; dx \in F_{p-1}(\vec{\mathcal{H}};\mathbb{K}) \,\big\}.
\end{equation}
\end{definition}
 
\begin{theorem}\label{thm_functor}
The family \(\Omega_{\ast}(\vec{\mathcal{H}};\mathbb{K}) = (\Omega_{p}(\vec{\mathcal{H}};\mathbb{K}))_{p\geq 0}\) is a chain complex, and it is the maximal chain complex embedded in
\(F_{\ast}(\vec{\mathcal{H}};\mathbb{K})\). Its homology \(H_{p}(\vec{\mathcal{H}};\mathbb{K}) = H_{p}\big(\Omega_{\ast}(\vec{\mathcal{H}};\mathbb{K})\big)\) is the embedded homology of \(\vec{\mathcal{H}}\), and \(\beta_{p} = \dim H_{p}(\vec{\mathcal{H}};\mathbb{K})\) is its \(p\)-th Betti number. Moreover, a morphism of hyperdigraphs induces a \(\mathbb{K}\)-linear map on embedded homology, so \(H_{\ast}(-;\mathbb{K})\) is functorial.
\end{theorem}
 
\begin{proof}
See \cite{bressan2019embedded} for the hypergraph case and \cite{chen2023persistent} for hyperdigraphs.
\end{proof}
 
We now take \(\mathbb{K} = \mathbb{R}\), so that the inner product required below is available, and equip $F_{\ast}(\vec{\mathcal{H}};\mathbb{R})$ with the inner product in which the directed hyperedges form an orthonormal basis. The subspace \(\Omega_{\ast}(\vec{\mathcal{H}};\mathbb{R})\) inherits this inner product, and the adjoint \((d_{p})^{\ast} : \Omega_{p-1} \to \Omega_{p}\) is defined by \(\langle d_{p}x, y\rangle = \langle x, (d_{p})^{\ast}y\rangle\).
 
\begin{definition}[Hyperdigraph Laplacian]\label{def_laplacian}
The \(p\)-th hyperdigraph Laplacian is the operator \(\Delta_{p}^{\vec{\mathcal{H}}} : \Omega_{p}(\vec{\mathcal{H}};\mathbb{R}) \to \Omega_{p}(\vec{\mathcal{H}};\mathbb{R})\) given by
\begin{equation}\label{eq_laplacian}
  \Delta_{p}^{\vec{\mathcal{H}}}
  = (d_{p})^{\ast}\circ d_{p}
  + d_{p+1}\circ (d_{p+1})^{\ast},
  \qquad p \geq 0.
\end{equation}
\end{definition}
 
Fix an orthonormal basis of \(\Omega_{\ast}(\vec{\mathcal{H}};\mathbb{R})\) and let \(B_{p}\) denote the matrix of \(d_{p}\) with respect to it. The matrix of \(\Delta_{p}^{\vec{\mathcal{H}}}\) is then \(L_{p}^{\vec{\mathcal{H}}} = B_{p}^{\top}B_{p} + B_{p+1}B_{p+1}^{\top}\), and we write \(n_{p} = \dim \Omega_{p}(\vec{\mathcal{H}};\mathbb{R})\). {Throughout, matrices act on column vectors of coordinates: \(B_{p}\in\mathbb{R}^{n_{p-1}\times n_{p}}\) has one row per basis element of \(\Omega_{p-1}\) and one column per basis element of \(\Omega_{p}\), so that \((d_{p})^{\ast}\) is represented by \(B_{p}^{\top}\) and \(L_{p}^{\vec{\mathcal{H}}}\in\mathbb{R}^{n_{p}\times n_{p}}\).}

{\begin{lemma}[Face-closed hyperdigraphs]\label{lem_face_closed}
Suppose that every vertex of \(V\) is a \(0\)-hyperedge of \(\vec{\mathcal{H}}\) and that, for every \(p\geq1\) and every \(x\in\vec{E}^{\,p}\), each face \(\partial_{i}x\), \(0\le i\le p\), belongs to \(\vec{E}^{\,p-1}\). Then \(\Omega_{p}(\vec{\mathcal{H}};\mathbb{R})=F_{p}(\vec{\mathcal{H}};\mathbb{R})\) for all \(p\), the directed hyperedges form an orthonormal basis of \(\Omega_{p}\), and \(B_{p}\) is the signed incidence matrix whose column indexed by \(x\in\vec{E}^{\,p}\) has the entry \((-1)^{i}\) in the row indexed by \(\partial_{i}x\) and zeros elsewhere. In particular, \(\operatorname{nnz}(B_{p})=(p+1)\,n_{p}\) for \(p\ge1\). Conversely, if some \(x\in\vec{E}^{\,p}\) has a face outside \(\vec{E}^{\,p-1}\), then \(x\notin\Omega_{p}\) and \(\Omega_{p}\ne F_{p}\).
\end{lemma}

\begin{proof}
Since the vertices of \(x=x_{0}\cdots x_{p}\) are distinct, the faces \(\partial_{0}x,\dots,\partial_{p}x\) have pairwise distinct vertex sets, so no two terms of \(d_{p}x=\sum_{i}(-1)^{i}\partial_{i}x\) cancel. Hence \(d_{p}x\in F_{p-1}\) if and only if every face of \(x\) lies in \(\vec{E}^{\,p-1}\). Under the hypothesis, \(d_{p}F_{p}\subseteq F_{p-1}\), so \(\Omega_{p}=F_{p}\) by \eqref{eq_omega}, with the orthonormal basis \(\vec{E}^{\,p}\) and the stated matrix. The converse is the same observation applied to \(x\).
\end{proof}

When the hypothesis of Lemma~\ref{lem_face_closed} fails, \(\Omega_{p}\) is a proper subspace of \(F_{p}\), and an orthonormal basis of it must be computed from \eqref{eq_omega} before the matrices \(B_{p}\) are available. That basis need not be sparse, and its construction is a setup cost separate from the per-probe cost below.}
 
\begin{theorem}\label{thm_hodge}
The operator \(\Delta_{p}^{\vec{\mathcal{H}}}\) is self-adjoint and non-negative definite, so its eigenvalues \(0 \leq \lambda_{1} \leq \cdots \leq \lambda_{n_{p}}\) are real and non-negative. Moreover,
\begin{equation}\label{eq_hodge}
  \Omega_{p}(\vec{\mathcal{H}};\mathbb{R})
  = \ker \Delta_{p}^{\vec{\mathcal{H}}}
    \oplus \operatorname{Im} d_{p+1}
    \oplus \operatorname{Im} (d_{p})^{\ast},
  \quad
  \ker \Delta_{p}^{\vec{\mathcal{H}}}
  = \ker d_{p} \cap \ker (d_{p+1})^{\ast}
  \cong H_{p}(\vec{\mathcal{H}};\mathbb{R}),
\end{equation}
 the multiplicity of the zero eigenvalue equals $\beta_{p}$, while the non-zero eigenvalues carry geometric information that the Betti numbers do not determine.
\end{theorem}
 
\begin{proof}
See \cite{chen2023persistent}; the underlying algebraic Hodge decomposition is due to \cite{eckmann1944harmonische}, and a modern account is given in \cite{lim2020hodge}.
\end{proof}
 
\subsection{Persistent Hyperdigraph}\label{persistent_hyperdigraph}
To describe a hyperdigraph across scales, we let its hyperedge set grow with a parameter and track the Laplacians along the resulting family, following \cite{chen2023persistent}.

\begin{definition}[Persistence hyperdigraph]\label{def_persistence}
Let \((\mathbb{R},\leq)\) denote the category whose objects are the real numbers and whose morphisms are \(a \to b\) for \(a \leq b\), and let \(\vec{\mathbf{H}}\mathbf{yper}\) denote the category of hyperdigraphs. A persistence hyperdigraph is a functor \(\mathcal{P} : (\mathbb{R},\leq) \to \vec{\mathbf{H}}\mathbf{yper}\). For \(a \leq b\), the \((a,b)\)-persistent homology of \(\mathcal{P}\) is
\begin{equation}\label{eq_persistent_homology}
  H_{p}^{a,b}(\mathcal{P};\mathbb{K})
  = \Ima\big( H_{p}(\mathcal{P}(a);\mathbb{K})
    \to H_{p}(\mathcal{P}(b);\mathbb{K}) \big),
\end{equation}
the image of the map induced on embedded homology by the morphism \(a \to b\), and the \((a,b)\)-persistent Betti number is \(\beta_{p}^{a,b} = \dim H_{p}^{a,b}(\mathcal{P};\mathbb{R})\) for \(p \geq 0\).
\end{definition}

Definition~\ref{def_persistence} is well posed precisely because embedded homology is functorial. For the Laplacian, we restrict to the subcategory \(\vec{\mathbf{H}}\mathbf{yper}^{\hookrightarrow}\) whose morphisms are inclusions of hyperdigraphs, so that each $a \leq b$ yields an inclusion \(j^{a,b} : \mathcal{P}(a) \hookrightarrow \mathcal{P}(b)\) and hence an inclusion of chain complexes \(\mathfrak{j}^{a,b}_{p} : \Omega_{p}^{a} \hookrightarrow \Omega_{p}^{b}\), where \(\Omega_{p}^{t} = \Omega_{p}(\mathcal{P}(t);\mathbb{R})\).
 
\begin{definition}[Persistent hyperdigraph Laplacian]
\label{def_persistent_laplacian}
For \(a \leq b\) set
\begin{equation}\label{eq_persistent_chains}
  \Omega_{p+1}^{a,b}
  = \big\{\, x \in \Omega_{p+1}^{b}
    \;:\; d_{p+1}^{b}x \in \Omega_{p}^{a} \,\big\},
  \qquad
  d_{p+1}^{a,b}
  = \big(\mathfrak{j}_{p}^{a,b}\big)^{\ast}
    \circ d_{p+1}^{b}
    \circ \iota_{p+1}^{a,b},
\end{equation}
where \(\iota_{p+1}^{a,b} : \Omega_{p+1}^{a,b} \hookrightarrow \Omega_{p+1}^{b}\) denotes the inclusion. The \(p\)-th
\((a,b)\)-persistent hyperdigraph Laplacian \(\Delta_{p}^{a,b} : \Omega_{p}^{a} \to \Omega_{p}^{a}\) is
\begin{equation}\label{eq_persistent_laplacian}
  \Delta_{p}^{a,b}
  = d_{p+1}^{a,b} \circ \big(d_{p+1}^{a,b}\big)^{\ast}
  + \big(d_{p}^{a}\big)^{\ast} \circ d_{p}^{a}.
\end{equation}
Fixing orthonormal bases of \(\Omega_{p-1}^{a}\), \(\Omega_{p}^{a}\) and \(\Omega_{p+1}^{a,b}\) and writing \(B_{p}^{a}\), \(B_{p+1}^{a,b}\) for the matrices of \(d_{p}^{a}\) and \(d_{p+1}^{a,b}\), with the column-vector convention above, so that \(B_{p}^{a}\in\mathbb{R}^{n_{p-1}^{a}\times n_{p}^{a}}\) and \(B_{p+1}^{a,b}\in\mathbb{R}^{n_{p}^{a}\times n_{p+1}^{a,b}}\), where \(n_{q}^{a}=\dim\Omega_{q}^{a}\) and \(n_{p+1}^{a,b}=\dim\Omega_{p+1}^{a,b}\), the matrix of \(\Delta_{p}^{a,b}\) is
\begin{equation}\label{eq_persistent_matrix}
  L_{p}^{a,b}
  ={B_{p+1}^{a,b}\big(B_{p+1}^{a,b}\big)^{\top}
  + \big(B_{p}^{a}\big)^{\top} B_{p}^{a}}
  \;{\in\mathbb{R}^{n_{p}^{a}\times n_{p}^{a}}},
\end{equation}
with spectrum \(\mathbf{Spec}(L_{p}^{a,b}) = \{\lambda_{p}^{a,b}(1),\dots,\lambda_{p}^{a,b}(n)\}\) arranged in ascending order, where \(n = \dim \Omega_{p}^{a}\). {For \(a=b\), \eqref{eq_persistent_matrix} reduces to \(L_{p}^{\vec{\mathcal{H}}}\) above.}
\end{definition}

\begin{theorem}[Persistent Hodge decomposition]
\label{thm_persistent_hodge}
For $a \leq b$,
\begin{equation}\label{eq_persistent_hodge}
  \Omega_{p}^{a}
  = \ker \Delta_{p}^{a,b}
  \oplus \Ima d_{p+1}^{a,b}
  \oplus \Ima \big(d_{p}^{a}\big)^{\ast},
  \qquad
  \ker \Delta_{p}^{a,b} \cong H_{p}^{a,b}(\mathcal{P};\mathbb{R}).
\end{equation}
\end{theorem}
 
\begin{proof}
See \cite{chen2023persistent}.
\end{proof}

Write \(\mathcal{H}_{p}^{a} = \ker \Delta_{p}^{a}\) with \(\Delta_{p}^{a}=\Delta_{p}^{a,a}\) for the harmonic space at parameter \(a\), and
\(\pi_{p}^{a} : \Omega_{p}^{a} \to \mathcal{H}_{p}^{a}\) for the associated orthogonal projection. The {\((a,b)\)-persistent harmonic space} is {the image of the earlier harmonic space,
\begin{equation}\label{eq_persistent_harmonic}
  \mathcal{H}_{p}^{a,b}
  = \Ima\Big(\big(\pi_{p}^{b} \circ \mathfrak{j}_{p}^{a,b}\big)\big|_{\mathcal{H}_{p}^{a}}\Big)
  = \pi_{p}^{b}\big(\mathfrak{j}_{p}^{a,b}(\mathcal{H}_{p}^{a})\big)
  \subseteq \mathcal{H}_{p}^{b}.
\end{equation}
The restriction to \(\mathcal{H}_{p}^{a}\) is essential. Without it, a chain at \(a\) that is not a cycle can have a nonzero projection onto a cycle born after \(a\). For example, two edges \(e_{1}=v_{0}v_{1}\) and \(e_{2}=v_{1}v_{2}\) at \(a\), joined at \(b\) by \(e_{3}=v_{0}v_{2}\) with no \(2\)-hyperedges, give \(\mathcal{H}_{1}^{b}=\operatorname{span}(e_{1}+e_{2}-e_{3})\), onto which \(e_{1}\) projects nontrivially, although \(H_{1}^{a}=0\).}
 
\begin{theorem}\label{thm_harmonic}
For all $a \leq b$ and $p \geq 0$,
\begin{equation}
  \mathcal{H}_{p}^{a,b} \cong H_{p}^{a,b}(\mathcal{P};\mathbb{R}),
\end{equation}
 the harmonic spaces and the homology carry the same persistence.
\end{theorem}
 
\begin{proof}
See \cite{chen2023persistent}.
%  {In brief, by Theorem~\ref{thm_hodge} \(h\mapsto[h]\) is an isomorphism \(\mathcal{H}_{p}^{t}\to H_{p}(\mathcal{P}(t);\mathbb{R})\), and \(\ker d_{p}^{b}=\mathcal{H}_{p}^{b}\oplus\Ima d_{p+1}^{b}\). For \(h\in\mathcal{H}_{p}^{a}\), the chain \(\mathfrak{j}_{p}^{a,b}h\) is a cycle at \(b\), so it differs from \(\pi_{p}^{b}\mathfrak{j}_{p}^{a,b}h\) by a boundary. Hence \(\pi_{p}^{b}\mathfrak{j}_{p}^{a,b}h\) is the harmonic representative of the image of \([h]\) under the map induced by \(a\to b\). The restricted map in \eqref{eq_persistent_harmonic} therefore corresponds to this induced map, and its image is isomorphic to \(H_{p}^{a,b}(\mathcal{P};\mathbb{R})\).}
\end{proof}
 
\begin{remark}[Distance-based filtration]\label{rem_filtration}
The filtration used in this work is distance-based, in the manner of a Vietoris–Rips construction. Let $\vec{\mathcal{H}} = (V,\vec{E})$ be a hyperdigraph whose vertices are points in Euclidean space. For
\(a \in \mathbb{R}\), put \(\mathcal{P}(a) = (V, \vec{E}(a))\) with
\begin{equation}\label{eq_distance_filtration}
  \vec{E}(a) = \big\{\, x \in \vec{E} \;:\;
  \text{every pair of points in } x
  \text{ lies within distance } a \,\big\}.
\end{equation}
Then \(a \mapsto \mathcal{P}(a)\) is a functor \((\mathbb{R},\leq) \to \vec{\mathbf{H}}\mathbf{yper}^{\hookrightarrow}\),
so Definitions~\ref{def_persistence} and \ref{def_persistent_laplacian} apply, and the geometry of the interface
enters only through \eqref{eq_distance_filtration}.
\end{remark}

% \updatecomment{The persistent harmonic space is now defined as the image of $\mathcal H_p^a$ under $\pi_p^b\circ\mathfrak j_p^{a,b}$, matching the source restriction in the cited hyperdigraph paper. A counterexample shows why the restriction is needed, and a short argument for Theorem~\ref{thm_harmonic} has been added (blue).}

\subsection{Stochastic Trace Estimation}\label{stochastic_trace}

The persistent Laplacians of the previous subsection act on spaces $\Omega_{p}^{a}$ whose dimension grows combinatorially in both the number of atoms and the order $p$, since a directed $p$-hyperedge is an ordered
sequence of $p+1$ distinct vertices. Extracting their spectra by diagonalization costs $\mathcal{O}(n^{3})$ operations and $\mathcal{O}(n^{2})$ storage per matrix, repeated for every order and every filtration pair. For protein–protein interfaces, this becomes prohibitive beyond the first few orders, and existing applications have accordingly remained there.

This subsection develops an alternative. Rather than the eigenvalues themselves, we use the spectral moments as descriptors.
\begin{equation}\label{eq_moments}
  \operatorname{tr}\big(L_{p}^{a,b}\big)=\sum_{i}\lambda_{i},
  \qquad
  \big\lVert L_{p}^{a,b}\big\rVert_{F}^{2}=\sum_{i}\lambda_{i}^{2},
\end{equation}
where $n=\dim\Omega_{p}^{a}$ and $0\le\lambda_{1}\le\cdots\le\lambda_{n}$ are the eigenvalues of
$L_{p}^{a,b}$. Three things must be established: that both moments are obtainable without forming the matrix whose spectrum they summarize; that the accuracy of the resulting estimates does not degrade as $n$ grows; and that the number of probes needed can be certified in advance from the boundary matrices alone.

\begin{theorem}[Unbiased matrix-free estimation]\label{thm_hutchinson}
Let \(z\) be a Rademacher vector, with independent entries taking the values $\pm1$ with equal probability, and let \(B_{p}^{a}\) and \(B_{p+1}^{a,b}\) be the boundary matrices of Definition~\ref{def_persistent_laplacian}. Then
$\mathbb{E}[zz^{\top}]=I_{n}$ and
\begin{equation}\label{eq_hutchinson}
  \mathbb{E}\big[z^{\top}L_{p}^{a,b}z\big]
  =\operatorname{tr}\big(L_{p}^{a,b}\big),
  \qquad
  \operatorname{Var}\big(z^{\top}L_{p}^{a,b}z\big)
  =2\Big(\big\lVert L_{p}^{a,b}\big\rVert_{F}^{2}
    -\sum_{i}\big(L_{p}^{a,b}\big)_{ii}^{2}\Big)
  \;\le\;2\big\lVert L_{p}^{a,b}\big\rVert_{F}^{2},
\end{equation}
the bound on the right being the variance obtained with Gaussian probes. Both moments in \eqref{eq_moments} are therefore estimated from the distribution of a single quadratic form. That form admits the factorization
\begin{equation}\label{eq_quadratic_form}
  z^{\top}L_{p}^{a,b}z
  ={\big\lVert \big(B_{p+1}^{a,b}\big)^{\top}z\big\rVert_{2}^{2}
  +\big\lVert B_{p}^{a}z\big\rVert_{2}^{2}}.
\end{equation}
Each evaluation costs two sparse matrix–vector products and \(\mathcal{O}(N+n)\) operations, where
\(N=\operatorname{nnz}(B_{p+1}^{a,b})+\operatorname{nnz}(B_{p}^{a})\), with storage proportional to $N$ rather than to $n^{2}${, once \(B_{p}^{a}\) and \(B_{p+1}^{a,b}\) are available. In the face-closed case of Lemma~\ref{lem_face_closed} these are the raw boundary matrices and require no further setup}.
\end{theorem}

\begin{proof}
See Section~\ref{proof_thm_hutchinson} of the supplementary material.
\end{proof}

Rademacher probes are used throughout this work. They satisfy the same isotropy condition, have strictly smaller variance whenever \(L_p^{a,b}\)  has a non-zero diagonal, and every bound below is stated in terms of 
\(\big\lVert L_p^{a,b}\big\rVert_F^2\)	and therefore holds for normal distribution as well.

\begin{remark}\label{rem_fillin}
The savings exceed the avoided $\mathcal{O}(n^{3})$ diagonalization.
Forming $L_{p}^{a,b}$ at all requires the product
{$B_{p+1}^{a,b}(B_{p+1}^{a,b})^{\top}$}, which is generally far denser than
its factor, since any two rows sharing a column may contribute a non-zero
entry. The factorization \eqref{eq_quadratic_form} avoids this fill-in as
well.
\end{remark}

The factorization also settles what the estimator is measuring. Both
terms on the right of \eqref{eq_quadratic_form} are squared norms of
images of $z$ under the boundary operators, so the quadratic form is a
Dirichlet energy on chains.

\begin{proposition}[Harmonic annihilation]\label{prop_annihilate}
For $z\in\Omega_{p}^{a}$, the expression \eqref{eq_quadratic_form}
vanishes if and only if $z\in\ker L_{p}^{a,b}$. Writing
$z=z_{\mathrm{harm}}+z_{\perp}$ according to the Hodge decomposition of
Theorem~\ref{thm_persistent_hodge}, one has
$z^{\top}L_{p}^{a,b}z=z_{\perp}^{\top}L_{p}^{a,b}z_{\perp}$.
\end{proposition}

\begin{proof}
See Section~\ref{proof_prop_annihilate} of the supplementary material.
\end{proof}

\begin{corollary}[Insensitivity to zero modes]
\label{cor_invisible}
For any $f$ with $f(0)=0$, the quantity
$\operatorname{tr}\big(f(L_{p}^{a,b})\big)$ {depends only on the non-zero eigenvalues of $L_{p}^{a,b}$ and not on the multiplicity $\beta_{p}^{a,b}$ of the zero eigenvalue}: two filtration pairs whose non-zero spectra agree
return the same value, although their persistent Betti numbers may differ.
In particular, the descriptors \eqref{eq_moments} are insensitive to additional zero modes, and are complementary to the invariants of
Theorem~\ref{thm_persistent_hodge} rather than a substitute for them. %{This is an algebraic statement about individual operators; it does not imply that the descriptors are statistically independent of topology across data.}
\end{corollary}

\begin{proof}
See Section~\ref{proof_cor_invisible} of the supplementary material.
\end{proof}

\begin{remark}
\label{rem_betti_probe}
The complementary quantity is also an expectation: if $\Pi_{0}$ denotes
the orthogonal projector onto $\ker L_{p}^{a,b}$, then
$\mathbb{E}\big[\lVert\Pi_{0}z\rVert^{2}\big]
=\operatorname{tr}(\Pi_{0})=\beta_{p}^{a,b}$, with variance
$2\beta_{p}^{a,b}$ for Gaussian $z$. For $\beta_{p}^{a,b}>0$, averaging
$s$ independent Gaussian probes gives relative root mean square error
$\sqrt{2/(s\beta_{p}^{a,b})}$. If $\beta_{p}^{a,b}=0$, the projector
and every probe value are zero. Applying
$\Pi_{0}$ requires a nullspace computation and is not pursued here.
\end{remark}

Two consequences of \eqref{eq_quadratic_form} make the first moment
computable exactly and provide a check on the estimates.

\begin{lemma}[The relation with the boundary matrices]
\label{lem_trace_frob}
\begin{equation}\label{eq_trace_frob}
  \operatorname{tr}\big(L_{p}^{a,b}\big)
  =\big\lVert B_{p+1}^{a,b}\big\rVert_{F}^{2}
  +\big\lVert B_{p}^{a}\big\rVert_{F}^{2},
\end{equation}
computable in $\mathcal{O}(N)$ time from the stored entries, with no
probing and no assembly.
\end{lemma}

\begin{proof}
$\operatorname{tr}(B^{\top}B)=\operatorname{tr}(BB^{\top})
=\lVert B\rVert_{F}^{2}$ applied to each summand of
$L_{p}^{a,b}=B_{p+1}^{a,b}(B_{p+1}^{a,b})^{\top}
+(B_{p}^{a})^{\top}B_{p}^{a}$.
\end{proof}

\begin{corollary}[Alternating trace identity]\label{cor_alternating}
At a fixed filtration value $a=b$, writing
$f_{p}=\lVert B_{p}^{a}\rVert_{F}^{2}$, so that $f_{0}=0$ because $d_{0}=0$, one has for every $P\ge0$
\begin{equation}\label{eq_alternating}
  \sum_{p=0}^{P}(-1)^{p}\operatorname{tr}\big(L_{p}(a)\big)=(-1)^{P}f_{P+1} .
\end{equation}
In particular, the sum vanishes when the construction is truncated at order $P$, so that no $(P+1)$-hyperedges enter $L_{P}(a)$.
\end{corollary}

\begin{proof}
By Lemma~\ref{lem_trace_frob}, $\operatorname{tr}(L_{p}(a))=f_{p}+f_{p+1}$,
and the alternating sum telescopes to $f_{0}-(-1)^{P+1}f_{P+1}=(-1)^{P}f_{P+1}$.
\end{proof}

\begin{remark}[Numerical consistency check]
\label{rem_check}
The alternating-trace identity \eqref{eq_alternating} provides a consistency check for the probe-based trace estimates across the retained orders. In this work, we retain orders up to the largest feasible order $P$ to capture higher-order topological information. Accordingly, the alternating sum is determined by the order-$(P+1)$ boundary contribution
and is not generally zero.
\end{remark}
% \begin{remark}
% \label{rem_check}
% Identity \eqref{eq_alternating} is an exact linear constraint on
% quantities the estimator returns across the six orders used in
% Section~\ref{hyperdigraph_construction}.\revised{ If the probes for different orders are drawn independently, the statistic
% $\sum_{p}(-1)^{p}\widehat{\operatorname{tr}}(L_{p}(a))$ has mean $(-1)^{P}\lVert B_{P+1}^{a}\rVert_{F}^{2}$, which is zero under truncation at order $P$, and
% variance $\sum_{p}\operatorname{Var}_{p}$. A deviation from this mean flags
% either too few probes or an error in the construction of the boundary
% matrices. Both sides are computable exactly in $\mathcal{O}(N)$ time by Lemma~\ref{lem_trace_frob}, so this is a consistency check on the implementation rather than a substitute for exact computation. The identity involves only traces and carries no homological information}, so it is a numerical check rather than a
% topological one.
% \end{remark}

Theorem~\ref{thm_hutchinson} makes each probe cheap, but the estimator it
defines is random. The next result bounds the number of probes required
for a prescribed accuracy, first abstractly and then in terms of
quantities available before any spectral computation.

\begin{theorem}[Probe complexity]\label{thm_recovery}
Write $L=L_{p}^{a,b}$, {a fixed positive semidefinite matrix with $L\ne0$,} let $z_{1},\dots,z_{s}$ be independent probe vectors, each with independent Rademacher entries, and set
$T_{s}=s^{-1}\sum_{j=1}^{s}z_{j}^{\top}Lz_{j}$. Then $T_{s}$ is unbiased
for $\operatorname{tr}(L)$, and for any $\varepsilon,\delta\in(0,1)$ and $s\;\ge\;\frac{2}{\delta\varepsilon^{2}}$,
\begin{equation}\label{eq_probe_count}
  \mathbb{P}\Big(\big|T_{s}-\operatorname{tr}(L)\big|
  \ge\varepsilon\operatorname{tr}(L)\Big)\;\le\;\delta .
\end{equation}
The number of probes needed for a prescribed relative accuracy of the trace estimate is
therefore independent of $n$.
\end{theorem}

\begin{proof}
See Section~\ref{proof_thm_recovery} of the supplementary material.
\end{proof}

\begin{remark}[Scope of the guarantee]
\label{rem_recovery_scope}
Theorem~\ref{thm_recovery} gives a concentration guaranty for the trace estimate of a single fixed nonzero operator. For the zero operator, the trace is recovered exactly since every probe value vanishes. The theorem does not provide the same guaranty for the other probe-based statistics
used below. For simultaneous control of $K$ nonzero operators, the same result applies with a failure probability $\delta/K$ for each operator, by the union bound.
\end{remark}

The bound \eqref{eq_probe_count} is uniform over all operators and therefore pessimistic for any particular one. The governing quantity is the spectral participation ratio, defined for $L\ne0$ by
\begin{equation}\label{eq_stable_rank}
  r(L)=\frac{\operatorname{tr}(L)^{2}}{\lVert L\rVert_{F}^{2}}
      =\frac{\big(\sum_{i}\lambda_{i}\big)^{2}}{\sum_{i}\lambda_{i}^{2}}
  \;\in\;[1,n],
\end{equation}
in terms of which independent Rademacher probes satisfy
\[
  \frac{\operatorname{Var}(T_s)}{\operatorname{tr}(L)^2}
  = \frac{2\big(\|L\|_F^2-\sum_i L_{ii}^2\big)}
         {s\,\operatorname{tr}(L)^2}
  \le \frac{2}{s\,r(L)}.
\]
For independent standard Gaussian probes, the relative variance is exactly $2/(s\,r(L))$. The following theorem brackets $r(L)$ between quantities that require no eigenvalue.

\begin{theorem}[Sandwich]\label{thm_sandwich}
Let $L=L_{p}^{a,b}\ne0$ and put
\begin{equation}\label{eq_schur}
  \Lambda_{p}
  =\sigma\big(B_{p+1}^{a,b}\big)^{2}+\sigma\big(B_{p}^{a}\big)^{2},
  \qquad
  \sigma(B)=\sqrt{\lVert B\rVert_{1}\,\lVert B\rVert_{\infty}},
\end{equation}
where $\lVert\cdot\rVert_{1}$ and $\lVert\cdot\rVert_{\infty}$ are the
maximum absolute column and row sums. Then
\begin{equation}\label{eq_sandwich}
  \frac{\lVert B_{p+1}^{a,b}\rVert_{F}^{2}
        +\lVert B_{p}^{a}\rVert_{F}^{2}}{\Lambda_{p}}
  \;\le\;\frac{\operatorname{tr}(L)}{\lambda_{n}}
  \;\le\;r(L)\;\le\;n-\beta_{p}^{a,b},
\end{equation}
with equality on the right exactly when the non-zero spectrum is flat. Every quantity in the leftmost expression is computable in $\mathcal{O}(N)$ time from the stored entries of the boundary matrices.
\end{theorem}

\begin{proof}
See Section~\ref{proof_thm_sandwich} of the supplementary material.
\end{proof}

\begin{corollary}[A priori probe count]\label{cor_apriori}
For $\varepsilon,\delta\in(0,1)$,
\begin{equation}\label{eq_apriori}
  s\;\ge\;\frac{2\,\Lambda_{p}}
  {\delta\,\varepsilon^{2}\big(\lVert B_{p+1}^{a,b}\rVert_{F}^{2}
  +\lVert B_{p}^{a}\rVert_{F}^{2}\big)}
\end{equation}
suffices for the conclusion of Theorem~\ref{thm_recovery} {under its hypotheses}. The right-hand side depends only on the stored entries of the boundary matrices, so the probe count can be certified before any spectral computation and separately for each element pair, cutoff, and order.
\end{corollary}

\begin{proof}
See Section~\ref{proof_cor_apriori}  of the supplementary material.
\end{proof}

\begin{remark}\label{rem_two_roles}
For independent Rademacher or standard Gaussian probes and $L\ne0$, the leftmost bound in \eqref{eq_sandwich} gives
\[
  \frac{\operatorname{Var}(T_s)}{\operatorname{tr}(L)^2}
  \le \frac{2}{s\,r(L)}
  \le \frac{2\Lambda_p}{s\,\operatorname{tr}(L)}.
\]
The final expression is computable from the boundary-matrix entries and provides the variance upper bound used in Corollary~\ref{cor_apriori}. The rightmost inequality, $r(L)\le n-\beta_p^{a,b}$, bounds the participation ratio by the number of nonzero eigenvalues; it does not provide a variance lower bound for Rademacher probes. For example, a nonzero diagonal $L$ has a constant Rademacher quadratic form and hence zero estimator variance.
\end{remark}

\begin{remark}\label{rem_moments}
The sample \(\{z_j^\top L_p^{a,b}z_j\}_{j\le s}\) carries more than its mean. By \eqref{eq_hutchinson},its sample variance estimates~\(
2\left(\|L_p^{a,b}\|_F^2-\sum_i\bigl((L_p^{a,b})_{ii}\bigr)^2\right).\) The empirical variance used in our descriptor is the uncorrected sample dispersion of these probe values. Its expectation is therefore the quantity above multiplied by the usual finite-sample factor $(s-1)/s$. We do not remove this factor or apply the diagonal correction;
instead, the resulting quantity is retained as a probe-dispersion feature.
Together with the first moment, these quantities describe spectral information beyond the kernel alone. The first moment is available exactly and without probing by Lemma \ref{lem_trace_frob}, and the diagonal correction is likewise computable directly from the boundary operators. Probing is used to estimate the second moment and obtain the order statistics used in Section \ref{hyperdigraph_construction}.
\end{remark}

\section{Persistent Hyperdigraph Learning Framework}\label{MFPHL_framework}

This section describes how the constructions of Section~\ref{Methods} are assembled into the PHL and MFPHL predictors through four stages, summarized in Figure~\ref{fig:workflow}. Starting from a protein–protein complex, interface atoms are partitioned by element into ordered channels, and a distance filtration generates directed hyperedges oriented by electronegativity. Two element-specific constructions are considered: an ordinary hyperdigraph incorporating multiple topological orders to capture many-body geometry and an order-one bipartite hyperdigraph restricted to cross-interface contacts to capture pairwise interactions. Both give boundary matrices, from which the two predictors diverge: PHL assembles each Laplacian and diagonalizes it, summarizing the resulting spectrum, while MFPHL leaves the Laplacian unformed and summarizes instead a sample of quadratic forms obtained by probing the boundary matrices. The two paths produce descriptors of identical shape, so they enter the model interchangeably. In both cases the features are combined block by block with ESM-2 embeddings into a fixed-length vector used to train a gradient boosting decision tree for binding affinity prediction.
\begin{figure}[!htbp]
    \centering
    \includegraphics[width=1\linewidth]{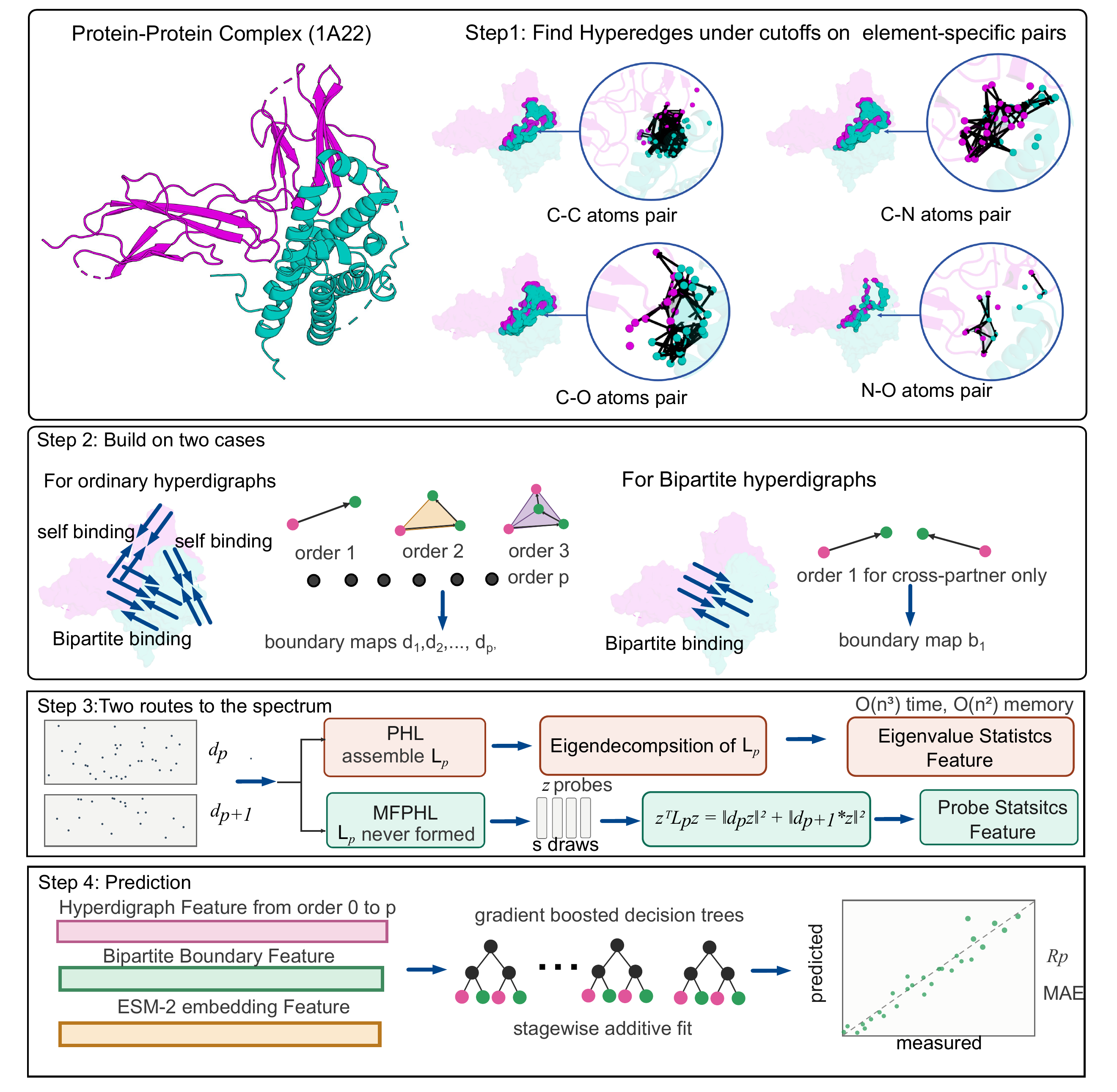}
    \caption{Overview of the proposed PHL and MFPHL frameworks for protein-protein binding affinity prediction. The workflow comprises four steps: element-specific hyperedges under distance cutoffs;
  hyperdigraph and bipartite constructions with their boundary maps; two routes to the spectrum, by eigendecomposition of the assembled Laplacian for PHL or by randomized probing of the boundary matrices for
  MFPHL and feature alignment with ESM-2 embeddings and prediction by
  gradient boosted trees.}
    \label{fig:workflow}
\end{figure}

\subsection{Data Sources and Preprocessing}

The benchmark used here is P2P~\cite{xu2025plnet}, which is constructed from PDBbind V2020 and SKEMPI v2. PDBbind V2020 provides the original protein-protein structures and affinity records, while the curation and
preprocessing protocol used here follows the P2P construction summarized in \cite{xu2025plnet}. Starting from the $2{,}852$ protein-protein complexes in PDBbind V2020, this procedure verifies binding partners
against the RCSB PDB, retains affinities reported as $K_d$, $K_i$, or $\Delta G_{\mathrm{bind}}$, and removes incomplete structures or imprecise affinity records, leaving $2{,}571$ complexes. Reported $K_d$
and $K_i$ values are converted to molar units, while $\Delta G_{\mathrm{bind}}$ values are converted to the corresponding equilibrium constants; all affinities are then expressed on the same negative base-$10$ logarithmic scale. After removing $203$ complexes overlapping with the SKEMPI v2 wild-type set, $2{,}368$ complexes remain, which we denote by V2020.

From SKEMPI v2, we retain $343$ wild-type complexes, denoted WT. Their union with V2020 gives P2P$_{\mathrm{WT}}$ with $2{,}711$ complexes, which is a subset of the broader P2P benchmark~\cite{xu2025plnet}. Each entry includes the assignment of chains to
partners A and B, which defines the interface as contacts between atoms belonging to different partners. Structures are completed with Profix~\cite{xiang2001extending} and stripped of hydrogens, so all
subsequent constructions use heavy atoms only. Mutant records from SKEMPI v2 are excluded because a single substitution can leave the geometric representation nearly unchanged while substantially altering
the measured affinity; modeling such effects requires features defined relative to the wild-type parent and is left for future work.
\subsection{Persistent Hyperdigraph Construction}
\label{hyperdigraph_construction}

From each complex, we take the residues whose C$_{\alpha}$ atom lies within $12\,\text{\AA}$ of the opposite partner and form element-specific point clouds from their atoms. Hydrogens having been removed, the atom types considered are \(\mathcal{T}=\{\mathrm{S},\mathrm{C},\mathrm{N},\mathrm{O}\}\), giving \(16\) ordered element-pair channels. For \((\tau_{1},\tau_{2})\in\mathcal{T}\times\mathcal{T}\), the vertex set \(V\)
consists of the atoms of type \(\tau_{1}\) in partner A together with those
of type \(\tau_{2}\) in partner B.

At a cutoff \(\epsilon\), two vertices are adjacent when their Euclidean distance is at most \(\epsilon\), whether or not they belong to the same partner, so that intra-partner geometry is recorded alongside the interface. The directed \(p\)-hyperedges of \(\vec{E}(\epsilon)\) are the orderings of the \((p+1)\)-cliques of this adjacency that are non-decreasing in electronegativity,
\begin{equation}\label{eq_electronegativity}
  \mathrm{S}\,(2.44) \;\to\; \mathrm{C}\,(2.50) \;\to\;
  \mathrm{N}\,(3.07) \;\to\; \mathrm{O}\,(3.50),
\end{equation}
with all tied orderings retained. Thus, a contact between two atoms of the same element contributes two oppositely directed orderings. This yields a hyperdigraph \(\vec{\mathcal{H}}(\epsilon)=(V,\vec{E}(\epsilon))\) in the sense of Definition~\ref{def_hyperdigraph}, and
\(\epsilon\mapsto\vec{\mathcal{H}}(\epsilon)\) is a filtration in the sense
of Remark~\ref{rem_filtration}, evaluated at \(\epsilon\in\{3,\ldots,12\}\,\text{\AA}\). In addition, the clique-based construction is naturally closed under taking faces: deleting a vertex from an ordered clique produces a lower-order clique whose induced ordering remains non-decreasing in electronegativity.
Hence, whenever this face-closed structure is retained, $\Omega_p=F_p$ and the corresponding boundary operator reduce to the usual signed incidence matrix.

The number of cliques grows combinatorially with both the filtration cutoff $\epsilon$ and the order $p$, and a few dense element-pair channels can therefore account for most of the computational cost. To control this growth, we impose a cap of $1000$ on the raw directed hyperedges retained in each hyperdigraph dimension. For each element pair, cutoff, and order $p\geq1$, candidates are ranked by their filtration score—the edge length for $p=1$ and the largest pairwise distance in the clique for $p\geq2$ with deterministic index-based tie breaking, and at most the first $1000$ are retained. Consequently, above the vertex space $F_0$, the raw spaces $F_p$ from which the embedded chain groups $\Omega_p$ are constructed have dimension  at most $1000$. The cap acts on the raw hyperdigraph representation and may break the face closure of the original clique construction. When the retained hyperedges remain face-closed, $\Omega_p=F_p$. Otherwise, $\Omega_p\subseteq F_p$ retains the chains whose boundaries remain in the corresponding lower-order spaces, ensuring that $\partial\Omega_p\subseteq\Omega_{p-1}$. Thus, the cap controls the size of the represented hyperdigraph, while the embedded construction provides the valid chain complex on which the persistent hyperdigraph Laplacian is defined. For capped instances, the raw capped boundary matrices are not used directly. For capped instances, the raw boundary matrices are not used directly. In the exact PHL construction, the raw boundary is split into a
retained-face block $G_p$ and a leakage block $C_p$. An orthonormal basis $Q_p$ of $\ker(C_p)$ is computed by rank-revealing pivoted QR using a standard machine-precision-based rank tolerance, and the embedded boundary is formed as
\(B_p = Q_{p-1}^{\top}G_pQ_p .\) In MFPHL, the same embedded operator is evaluated matrix-free without forming $Q_p$ explicitly. The projection onto $\Omega_p$ is carried out
using LSMR~\cite{fong2011lsmr}, an iterative Krylov method for solving sparse least-squares problems, with tolerance $10^{-6}$ and at most $100$ iterations. This allows the embedded projection to be applied using sparse matrix-vector products rather than constructing a dense null-space basis.

Write \(L_{p}(\epsilon)=L_{p}^{\epsilon,\epsilon}\) for the matrix of the Laplacian of Definition~\ref{def_persistent_laplacian} acting on \(\Omega_{p}(\vec{\mathcal{H}}(\epsilon);\mathbb{R})\). For \(p=0,1,\ldots,5\), we draw \(s\) probe vectors \(z_{1},\ldots,z_{s}\) and
evaluate
\begin{equation}\label{eq_probe_sample}
  q_{j} = z_{j}^{\top}L_{p}(\epsilon)z_{j},
  \qquad j=1,\ldots,s,
\end{equation}
through the quadratic-form factorization \eqref{eq_quadratic_form}, avoiding explicit eigendecomposition. All descriptors in this work use only the diagonal pairs \(a=b=\epsilon\). They sample the hyperdigraph Laplacians \(L_{p}(\epsilon)\) along the filtration, and no cross-scale operator \(L_{p}^{a,b}\) with \(a<b\) is evaluated.

The sample $\{q_j\}_{j=1}^{s}$ is the object from which the MFPHL features are formed. In the reported experiments, we use independent Rademacher probes without normalization or rescaling and summarize the resulting probe values by eight statistics: sum, minimum, maximum, mean, standard deviation, uncorrected empirical variance, probe
value $\ell_2$ norm, and cardinality. Among these quantities, the sample mean directly estimates $\operatorname{tr}L_p(\epsilon)$ by Theorem~\ref{thm_hutchinson}. The remaining statistics characterize the
dispersion, scale, and extremes of the randomized probe responses and are therefore used as randomized operator descriptors rather than spectrum-only invariants. In particular, the empirical variance is used without a diagonal correction or additional scaling. The cardinality records the number of probes $s$, and the sum is correspondingly
redundant with the mean, but both are retained to match the feature set used in the reported experiments.
The final descriptor concatenates these eight statistics over the $16$ element-pair channels, the $10$ filtration cutoffs, and the six retained orders, yielding $16\times10\times6\times8=7680$ features. All reported probe-based
features are generated using the fixed random seed $20260714$.

\subsection{Directed Bipartite Boundary Features}
\label{sec:directed_bipartite_d1}

We complement the high-order features of Section~\ref{hyperdigraph_construction} with a first-order descriptor confined to the interface. The hyperdigraph construction admits contacts within a partner as well as across it; here, only cross-partner contacts are used, which we enforce through the modified distance
\begin{equation}\label{eq_modified_distance}
  D_{\mathrm{mod}}(a_{i},a_{j}) =
  \begin{cases}
    \lVert \mathbf{r}_{i} - \mathbf{r}_{j}\rVert,
      & {a_{i} \in A,\; a_{j} \in B \ \text{or}\ a_{i} \in B,\; a_{j} \in A},\\[2pt]
    \infty,
      & \text{otherwise},
  \end{cases}
\end{equation}
where $\mathbf{r}_{i}$ and $\mathbf{r}_{j}$ are the positions of the atoms \(a_{i}\) and \(a_{j}\). \(A\), \(B\) are the two partners. Adjacency at a cutoff $\epsilon$ means \(D_{\mathrm{mod}} \leq \epsilon\), which is finite only across the interface; thus, the adjacency is {symmetric and} bipartite by construction. Atom types, orientation, and cutoffs are as in
Section~\ref{hyperdigraph_construction}. {The orientation of each selected contact is assigned afterward by electronegativity and may point from $B$ to $A$}.

The resulting object is a special case of Definition~\ref{def_hyperdigraph} in which \(\vec{E}(\epsilon)\) contains
directed hyperedges of order at most one. Since there are no \(2\)-hyperedges, the second term of \eqref{eq_laplacian} vanishes, and the Laplacian at order one reduces to {\(M(\epsilon) = B_{1}(\epsilon)^{\top}B_{1}(\epsilon)\), where
\(B_{1}(\epsilon)\in\mathbb{R}^{|V|\times|\vec{E}^{\,1}(\epsilon)|}\) is the boundary matrix of Definition~\ref{def_persistent_laplacian}, in the column-vector convention of Section~\ref{hyperdigraph}, with one row per atom and one column per directed edge. Since every vertex belongs to \(F_{0}\), the boundary of each directed edge lies in \(F_{0}\); therefore, \(\Omega_{1}=F_{1}\), the directed edges form an orthonormal basis, and \(B_{1}(\epsilon)\) is the signed vertex–edge incidence matrix.}

Its quadratic form $z^{\top}M(\epsilon)z=\lVert B_{1}(\epsilon)z\rVert_{2}^{2}$ requires only a single sparse matrix-vector product, so the estimation procedure in Section~\ref{stochastic_trace} applies unchanged. We draw the same $s$ probe vectors, form the sample \eqref{eq_probe_sample}, and
summarize it using the same eight statistics. No additional reorientation of the directed edges is performed in the reported experiments. Concatenating these statistics over the $16$ element pairs
and $10$ cutoffs gives $16\times10\times8=1280$ additional features generated using random seed $20260714$.

\subsection{Protein Language Model Features}

Sequence embeddings from protein language models complement the structural descriptors above, since they encode evolutionary constraints that the geometry of a single crystal structure does not express. We use {ESM-2} \cite{lin2023evolutionary}, a transformer trained by masked language modeling on UniRef50, in the {\texttt{esm2\_t33\_650M\_UR50D} checkpoint} of $33$ layers and roughly $650$ million parameters, without fine-tuning. Each chain is embedded by averaging the per-residue representations of the final layer, resulting in a \(1{,}280\)-dimensional vector; chains longer than the model's context length are divided into
segments, each embedded in the same way and treated as separate units. The units of a complex are concatenated in the order fixed by the partner assignment and  truncated to three, giving a \(3,840\)-dimensional vector per complex.

\subsection{Machine Learning Parameters}
Binding-affinity prediction is performed using a Gradient Boosting Decision Tree (GBDT) model implemented  in \texttt{scikit-learn}. It builds an ensemble of regression trees sequentially, with each new tree fitted
to reduce the residual error of the current ensemble. We use the Huber loss with $\alpha=0.9$, $1200$ estimators, maximum tree depth $3$, minimum split size $4$, minimum leaf size $5$, learning rate $0.012$,
subsampling ratio $0.7$, and \texttt{max\_features = sqrt}. These hyperparameters were selected in a preliminary screen on the WT dataset using three seeds ($42$, $1234$, and $5678$) and were then fixed for all
reported datasets and cross-validation runs. Within each training fold, features and targets are standardized using training fold statistics. Feature selection is also performed within the fold, retaining the top $\min(1000,k)$ features, where $k$ is the number of available
features; the resulting transformations are then applied to the held-out fold. To account for stochasticity in GBDT training, the final evaluation is repeated over ten random seeds ($42$, $1234$, $5678$, $91011$, $121314$, $151617$, $181920$, $212223$, $242526$, and $272829$), and the reported performance is aggregated over these runs.

\section{Results}\label{results}
Predictive accuracy is measured by the Pearson correlation coefficient \(R_{p}\) and the mean absolute error as follows:
\begin{equation}\label{eq_metrics}
  R_{p} = \frac{\sum_{i=1}^{N}(y_{i}-\bar{y})(y_{i}^{p}-\bar{y}^{p})}
       {\sqrt{\sum_{i=1}^{N}(y_{i}-\bar{y})^{2}}
        \sqrt{\sum_{i=1}^{N}(y_{i}^{p}-\bar{y}^{p})^{2}}},
  \qquad
  \mathrm{MAE} = \frac{1}{N}\sum_{i=1}^{N}\big|y_{i}-y_{i}^{p}\big|,
\end{equation}
where \(N\) is the number of complexes, \(y_{i}\) and \(y_{i}^{p}\) are the experimental and predicted affinities of the \(i\)-th complex, and
\(\bar{y}\) and \(\bar{y}^{p}\) are their means. The first measures linear association between predicted and measured affinities; the second is absolute accuracy. Both are
computed under \(10\)-fold cross-validation on the combined set of PDBbind V2020 protein–protein complexes and the wild-type subset of SKEMPI v2.

Since the contribution is computational as well as predictive, we report the mean wall-clock time per complex, the total dataset time, and the corresponding \(R_p\) as the number of probe vectors in Theorem~\ref{thm_recovery} is varied. All timings are single-threaded.

For completeness, we also compare the proposed MFPHL estimator with a spectral PHL baseline computed from the exact eigenvalues of the same capped hyperdigraph operators. Both methods use the same raw hyperedge
construction, including the cap of $1000$ directed hyperedges per order and element-pair channel, and therefore target the same embedded order-$p$ Laplacian. Spectral PHL explicitly constructs this operator
and computes its eigenvalues, whereas MFPHL evaluates the corresponding quadratic forms through matrix-free projection. Here, exact refers to the eigendecomposition of this specified capped operator.

For each Laplacian order $p$, we first define the relative trace error for an individual Laplacian $L_p$ as
\[
\operatorname{Rel}(L_p)
=
\frac{
\left|
\frac{1}{s}\sum_{j=1}^{s} z_j^\top L_p z_j
-
\operatorname{tr}(L_p)
\right|
}{
|\operatorname{tr}(L_p)|+10^{-12}
}.
\]
The error reported below is the relative root-mean-square trace error,
computed separately for each order $p$:
\[
\operatorname{RelRMS}(L_p)
=
\sqrt{
\frac{1}{N_p}
\sum_{i=1}^{N_p}
\operatorname{Rel}\!\left(L_{p,i}\right)^2
},
\]
where $L_{p,i}$ denotes the $i$-th valid order-$p$ Laplacian, and $N_p$ is the number of operators included across PDB complexes, element-pair channels, and filtration cutoffs. 

The spectral comparison is limited only by the cost of dense eigendecomposition. Exact eigenvalues are computed when the dimension of the resulting embedded operator is at most $10^4$. Operators above
this threshold are not replaced by leading principal submatrices; they are omitted from the exact-eigenvalue comparison, and their number and fraction are reported separately. This threshold affects only the availability of the spectral baseline and does not modify the hyperdigraph operator used by MFPHL. Another separate size limit is used for the directed-bipartite baseline: for the C--C channel, at most $2.2\times10^4$ directed edges are retained. This limit applies only to the construction and is not used in the PHL and MFPHL comparison.

\subsection{Prediction Performance}
\label{sec:prediction}

We first evaluate predictive accuracy under $10$-fold cross-validation. Folds are drawn at random, and the same partition is reused across feature sets so that comparisons are paired. Table~\ref{tab:mfphl-performance} reports $R_p$ and MAE, while Figure~\ref{fig:performance} compares the results with previously reported methods and shows predicted versus measured affinities for PHL. PHL achieves $R_p=0.7337$, $0.6729$, and $0.6923$ on WT, V2020, and P2P$_{\mathrm{WT}}$, respectively, exceeding the reported PLNet, PML, and PLD-tree values where available~\cite{xu2025plnet,xu2026pml}. These literature values are included for reference and were not obtained by rerunning the comparator methods on the present cross-validation folds. For the scatter plots, the predictions for each complex are first averaged across the repeated random-seed runs, and the resulting out-of-fold predictions are then
aggregated across the cross-validation folds. The plots show that PHL captures the overall variation in measured affinities, with predictions tending toward the center of the observed range and larger deviations remaining near the extremes.

MFPHL retains much of this predictive performance at substantially lower feature generation cost, with an accuracy trade-off that depends on the dataset. On WT, its best correlation is \(0.7322\), only \(0.0015\) below PHL, and its MAE is also close to that of the spectral descriptor. This correlation difference is smaller
than the observed seed-to-seed variation, suggesting that the probe-based features preserve much of the information needed for prediction on WT{. Training-seed standard deviations alone, however, do not establish statistical equivalence}. On V2020 and P2P$_{\mathrm{WT}}$, however, PHL consistently achieves higher correlation and lower MAE, while MFPHL exceeds the reported PLNet correlations but falls below PLD-tree. {A plausible explanation is} that MFPHL uses moment estimates and probe statistics rather than the full set of eigenvalue-based descriptors: accurate estimation of these summaries does not necessarily preserve all predictively useful spectral information. Increasing the probe count produces no systematic improvement, and the accuracy gap persists on the two larger datasets throughout the tested range. Together with the runtime results in Table~\ref{tab:runtime}, these findings support MFPHL as a practical option when computational cost is a priority,
while PHL provides higher predictive accuracy when spectral computation is affordable.
\begin{figure}[!htbp]
    \centering
    \includegraphics[width=1\linewidth]{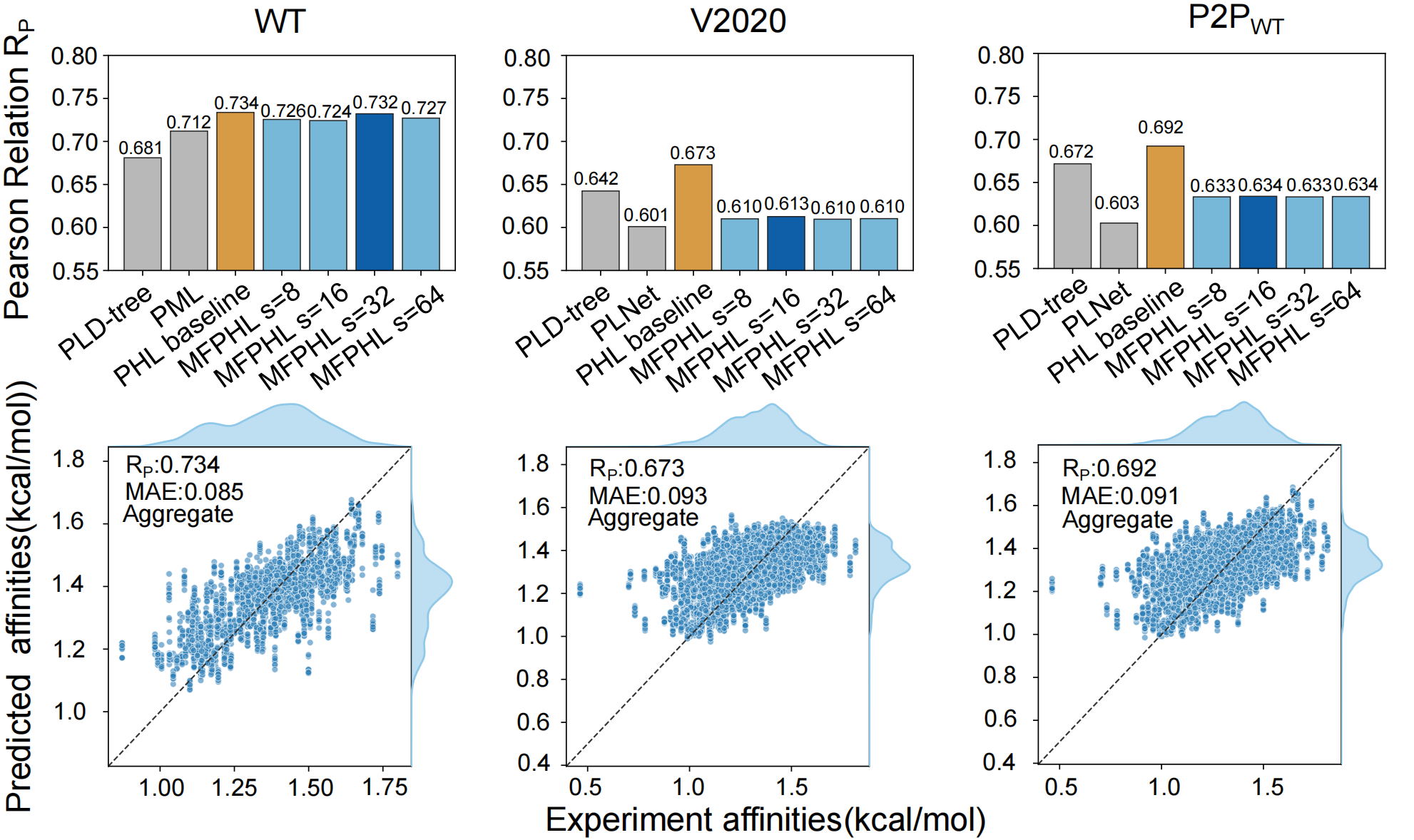}
    \caption{Prediction performance under \(10\)-fold cross-validation. Top: Pearson \(R_{p}\) on WT, V2020 and P2P\_WT for the probe descriptor at four probe counts, against the {eigenvalue-based PHL} baseline and the published PLNet~\cite{xu2025plnet}, PML~\cite{xu2026pml}{, and PLD-tree} results. Bottom:
    predicted versus measured binding affinities for PHL,
with marginal densities.}
    \label{fig:performance}
\end{figure}

\begin{table}[!htbp]
\centering
\small
\caption{Binding-affinity prediction performance, reported as mean \(\pm\) standard deviation over random seeds. The PHL baseline uses the eigenvalue-based persistent hyperdigraph descriptor; MFPHL uses the probe-based quadratic-form descriptor with \(s\) probe vectors.}
\label{tab:mfphl-performance}
\begin{tabular}{lccc}
\hline
Method & WT & V2020 & P2P$_{\mathrm{WT}}$ \\
\hline
\multicolumn{4}{l}{\text{Pearson} $R_p$} \\
PHL baseline  & 0.7337\,$\pm$\,0.0082 & 0.6729\,\(\pm\)0.0012                   & 0.6923\(\pm\)0.0010       \\
MFPHL, $s=8$  & 0.7256\,$\pm$\,0.0076 & 0.6101\,$\pm$\,0.0019 & 0.6333\,$\pm$\,0.0016 \\
MFPHL, $s=16$ & 0.7243\,$\pm$\,0.0061 & 0.6127\,$\pm$\,0.0023 & 0.6339\,$\pm$\,0.0018 \\
MFPHL, $s=32$ & 0.7322\,$\pm$\,0.0072 & 0.6096\,$\pm$\,0.0020 & 0.6332\,$\pm$\,0.0013 \\
MFPHL, $s=64$ & 0.7271\,$\pm$\,0.0044 & 0.6103\,$\pm$\,0.0021 & 0.6336\,$\pm$\,0.0021 \\
\hline
\multicolumn{4}{l}{\text{MAE}} \\
PHL baseline  & 0.0859\,$\pm$\,0.0010 & 0.0930 \(\pm\)0.0002                   & 0.0910\,\(\pm\) 0.0002                   \\
MFPHL, $s=8$  & 0.0872\,$\pm$\,0.0010 & 0.0964\,$\pm$\,0.0002 & 0.0948\,$\pm$\,0.0003 \\
MFPHL, $s=16$ & 0.0877\,$\pm$\,0.0006 & 0.0961\,$\pm$\,0.0003 & 0.0949\,$\pm$\,0.0003 \\
MFPHL, $s=32$ & 0.0869\,$\pm$\,0.0011 & 0.0963\,$\pm$\,0.0002 & 0.0949\,$\pm$\,0.0002 \\
MFPHL, $s=64$ & 0.0874\,$\pm$\,0.0006 & 0.0963\,$\pm$\,0.0002 & 0.0950\,$\pm$\,0.0003 \\
\hline
\end{tabular}
\end{table}

\subsection{Cost and Scaling of Feature Generation}\label{sec:runtime}

We next compare the computational cost of PHL and MFPHL using the same capped hyperdigraph construction. For each complex, element channel, and filtration cutoff, both methods begin from the same raw hyperdigraph and
target the same embedded operator. PHL then forms the corresponding Laplacian and computes eigenvalue-based statistics, whereas MFPHL replaces this eigendecomposition by $s$ Rademacher probe evaluations through sparse
boundary operations. All timings are single-threaded wall-clock measurements, with hyperdigraph construction and operator setup recorded separately from eigendecomposition or probe evaluation. ESM-2 embeddings
are precomputed and excluded from the timing comparison because they are shared by both methods.
Table~\ref{tab:runtime} reports the mean time per complex and total runtime on WT and V2020. On WT, the full PHL feature-generation pipeline requires approximately $363{,}497$ seconds, compared with $3{,}420$ seconds for MFPHL at $s=8$, corresponding to a speedup of about $106\times$; even at $s=64$, the speedup remains about $30\times$. On V2020, the corresponding speedups are approximately $121\times$ and $34\times$. These gains refer to replacing the full eigenvalue-statistics computation rather than the trace alone: by Lemma~\ref{lem_trace_frob}, the trace can already be computed exactly from the boundary representation without eigendecomposition. The exact-trace timing is therefore included as a separate baseline, while the main savings of MFPHL arise from avoiding the remaining spectral computations.

MFPHL runtime increases with the probe count because each additional probe requires further applications of the boundary operators. However, the total runtime grows more slowly than \(s\): an eightfold increase in probe count increases the overall cost by a factor of about \(3.6\) on both datasets, consistent with a combination of shared setup costs and probe-dependent work. The dependence also varies by topological order, with \(L_0\) timings remaining nearly constant while higher-order computations show a stronger increase. Figure~\ref{fig:scaling} illustrates the associated trade-off on WT: additional probes improve trace-estimation accuracy at greater computational cost. Since V2020 exhibits the same trend, only WT is shown. Together with Section~\ref{sec:prediction}, these results distinguish improved estimation accuracy from improved predictive accuracy: increasing \(s\) refines the trace estimates but does not systematically improve binding-affinity prediction. Small probe counts therefore capture much of MFPHL's practical computational advantage, while PHL provides the reference eigenvalue-based representation at substantially greater cost.

\begin{figure}[!htbp]
  \centering
  \includegraphics[width=0.9\linewidth]{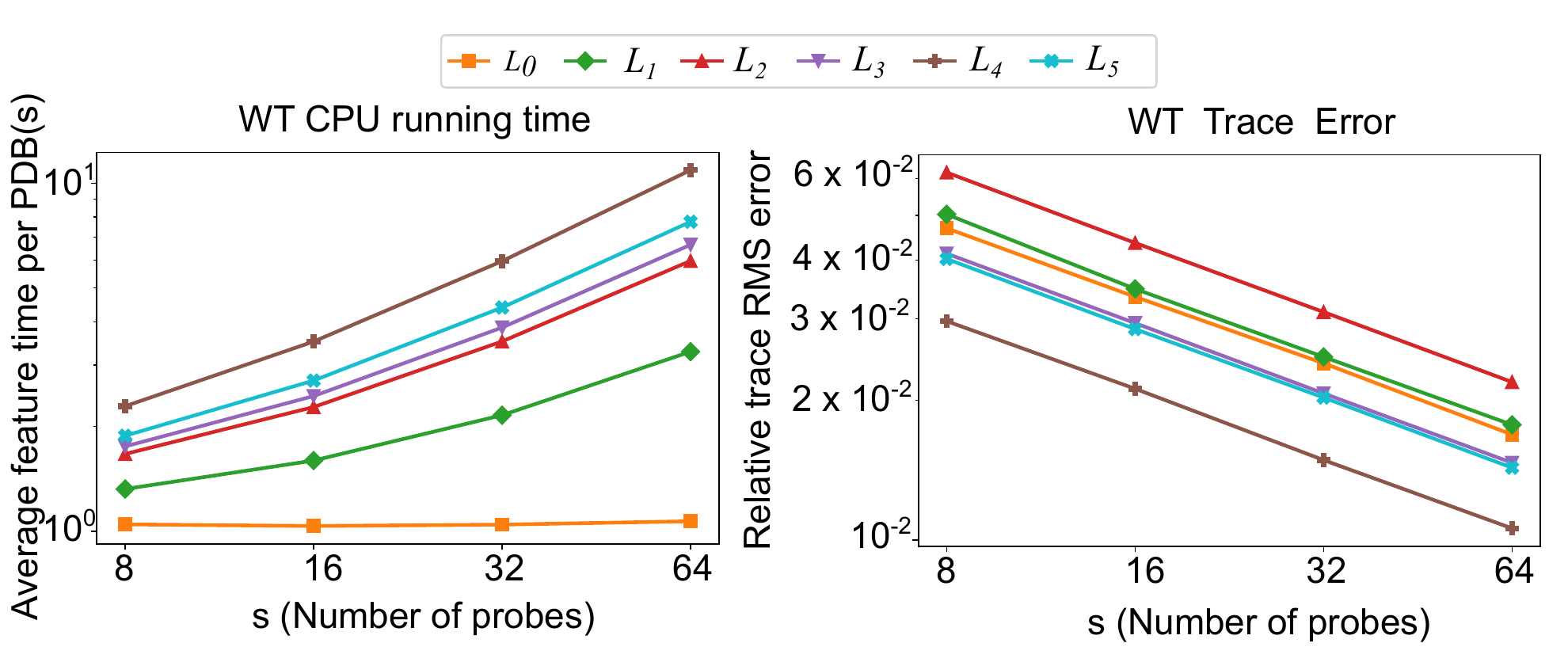}
  \caption{Scaling of feature-generation time and relative trace error with the number of probe vectors \(s\) on WT.}
  \label{fig:scaling}
\end{figure}

\begin{table}[!htbp] \centering \scriptsize \setlength{\tabcolsep}{2.2pt} \caption{Wall-clock time in seconds for persistent hyperdigraph feature generation. Here \(s\) denotes the number of probe vectors. MFPHL uses the quadratic-form estimator of Theorem~\ref{thm_hutchinson}, whereas the PHL baseline forms \(L_p(\epsilon)\) and computes the corresponding eigenvalue statistics.} \label{tab:runtime} 
\resizebox{\textwidth}{!}{% 
\begin{tabular}{lrrrrrrr} 
\hline 
Method & $b_{1}$ & $L_{0}$ & $L_{1}$ & $L_{2}$ & $L_{3}$ & $L_{4}$ & $L_{5}$ \\ 
\hline 
\multicolumn{8}{l}{\textbf{WT: mean per PDB complex}} \\ 
MFPHL, $s=8$ & 0.021 & 1.046 & 1.322 & 1.666 & 1.752 & 2.286 & 1.880 \\ 
MFPHL, $s=16$ & 0.028 & 1.035 & 1.595 & 2.273 & 2.455 & 3.507 & 2.709 \\ 
MFPHL, $s=32$ & 0.045 & 1.044 & 2.154 & 3.509 & 3.848 & 5.970 & 4.392 \\ 
MFPHL, $s=64$ & 0.076 & 1.068 & 3.280 & 5.980 & 6.653 & 10.905 & 7.746 \\ 
PHL baseline & 28.913 & 187.170 & 171.680 & 169.101 & 168.130 & 167.993 & 166.769 \\ 
\hline 
\multicolumn{8}{l}{\textbf{WT: total over all PDB complexes}} \\ 
MFPHL, $s=8$ & 7.044 & 358.725 & 453.342 & 571.577 & 600.888 & 784.084 & 644.833 \\ 
MFPHL, $s=16$ & 9.440 & 354.916 & 547.229 & 779.727 & 838.610 & 1,202.282 & 929.258 \\ 
MFPHL, $s=32$ & 15.404 & 358.183 & 738.840 & 1,203.702 & 1,319.872 & 2,047.744 & 1,506.523 \\ 
MFPHL, $s=64$ & 26.120 & 366.205 & 1,124.979 & 2,051.172 & 2,282.009 & 3,740.344 & 2,657.031 \\ 
PHL baseline & 9,917.244 & 64,199.264 & 58,886.346 & 58,001.621 & 57,668.750 & 57,621.534 & 57,201.765 \\ \hline 
\multicolumn{8}{l}{\textbf{V2020: mean per PDB complex}} \\ 
MFPHL, $s=8$ & 0.022 & 1.007 & 1.271 & 1.601 & 1.671 & 2.167 & 1.787 \\ 
MFPHL, $s=16$ & 0.030 & 0.995 & 1.542 & 2.195 & 2.334 & 3.324 & 2.573 \\ 
MFPHL, $s=32$ & 0.049 & 1.007 & 2.083 & 3.390 & 3.669 & 5.648 & 4.160 \\ 
MFPHL, $s=64$ & 0.084 & 1.030 & 3.174 & 5.781 & 6.336 & 10.295 & 7.325 \\ 
PHL baseline & 30.828 & 212.948 & 182.359 & 181.165 & 180.850 & 180.813 & 180.331 \\ 
\hline 
\multicolumn{8}{l}{\textbf{V2020: total over all PDB complexes}} \\ 
MFPHL, $s=8$ & 52.579 & 2,384.761 & 3,010.111 & 3,791.161 & 3,956.982 & 5,131.319 & 4,231.633 \\ 
MFPHL, $s=16$ & 70.838 & 2,356.952 & 3,650.636 & 5,197.592 & 5,527.054 & 7,871.015 & 6,091.745 \\ 
MFPHL, $s=32$ & 116.377 & 2,383.866 & 4,933.214 & 8,027.410 & 8,687.102 & 13,373.361 & 9,851.957 \\ 
MFPHL, $s=64$ & 199.235 & 2,438.224 & 7,515.032 & 13,690.547 & 15,002.696 & 24,379.209 & 17,344.979 \\ 
PHL baseline & 72,977.510 & 504,260.908 & 431,827.030 & 428,999.035 & 428,252.269 & 428,165.224 & 427,024.252 \\ 
\hline 
\end{tabular}% 
} \end{table}

\section{Discussion}\label{discussion}
The preceding results establish the predictive value of PHL and the computational advantages of MFPHL. We now examine
MFPHL in greater detail to understand its estimation behavior, the information retained by its descriptors, and the choices that guide their construction. Section~\ref{subsec:synthetic_hyperdigraph_speed} evaluates the accuracy-runtime trade-off on synthetic hyperdigraphs, where the exact trace is available. Section~\ref{subsec:why_works} explores the information captured by the two spectral moments and its relevance to prediction. 
Finally, Section~\ref{single-order} evaluates the contributions of individual topological orders to MFPHL prediction.

\subsection{Synthetic Benchmark}
\label{subsec:synthetic_hyperdigraph_speed}

We compare exact and matrix-free computations on synthetic directed hyperdigraphs to isolate the stochastic estimator from protein-specific effects. For each of ten random seeds, points are sampled uniformly from $[0,1]^3$ with $64$, $96$, or $128$ points and cutoff $0.30$, with no hyperedge cap. Each clique is directed by ordering its vertices according to projection onto $(1,1.5,2)/\lVert(1,1.5,2)\rVert_2$. Since every face of a retained clique is also retained, the construction is face-closed and $\Omega_p=F_p$, so the exact and matrix-free computations use the
same signed boundary operators without projection onto $\Omega_p$. We construct through order six so that $L_5$ includes the upper $B_6B_6^\top$ contribution.

Table~\ref{tab:synthetic-operators} summarizes the dimensions and sparsities of the resulting operators. The operator size grows substantially with the number of sampled points, with the largest spaces typically appearing at intermediate orders. In the face-closed setting, $\operatorname{nnz}(B_p)=(p+1)n_p$ for $p\geq1$. We further define $N_p=\operatorname{nnz}(B_p)+\operatorname{nnz}(B_{p+1})$ as the sparse
boundary work associated with one matrix-free probe. The table therefore shows both the growth of the exact operators and the corresponding matrix-free workload across orders.

\begin{table}[!htbp]
\centering
\small
\caption{Synthetic face-closed hyperdigraph operators, averaged over ten
point clouds. No hyperedge cap is applied. Here
$N_p=\operatorname{nnz}(B_p)+\operatorname{nnz}(B_{p+1})$.}
\label{tab:synthetic-operators}
\begin{tabular}{rrrrrr}
\hline
Points & Order & $n_p$ & $\operatorname{nnz}(B_p)$ &
$\operatorname{nnz}(L_p)$ & $N_p$ \\
\hline
64  & $L_0$ & 64.0   & 0.0    & 414.7  & 351.6 \\
64  & $L_1$ & 175.8  & 351.6  & 1169.4 & 907.2 \\
64  & $L_2$ & 185.2  & 555.6  & 744.0  & 998.8 \\
64  & $L_3$ & 110.8  & 443.2  & 339.2  & 649.2 \\
64  & $L_4$ & 41.2   & 206.0  & 100.4  & 258.2 \\
64  & $L_5$ & 8.7    & 52.2   & 16.5   & 57.8 \\
\hline
96  & $L_0$ & 96.0   & 0.0    & 826.2  & 730.4 \\
96  & $L_1$ & 365.2  & 730.4  & 3220.2 & 2313.5 \\
96  & $L_2$ & 527.7  & 1583.1 & 2765.7 & 3290.7 \\
96  & $L_3$ & 426.9  & 1707.6 & 1672.1 & 2808.6 \\
96  & $L_4$ & 220.2  & 1101.0 & 699.2  & 1555.2 \\
96  & $L_5$ & 75.7   & 454.2  & 200.5  & 573.9 \\
\hline
128 & $L_0$ & 128.0  & 0.0    & 1403.9 & 1276.4 \\
128 & $L_1$ & 638.2  & 1276.4 & 7143.4 & 4842.8 \\
128 & $L_2$ & 1188.8 & 3566.4 & 7630.6 & 8494.0 \\
128 & $L_3$ & 1231.9 & 4927.6 & 5727.1 & 9155.1 \\
128 & $L_4$ & 845.5  & 4227.5 & 3159.1 & 6767.9 \\
128 & $L_5$ & 423.4  & 2540.4 & 1342.6 & 3649.2 \\
\hline
\end{tabular}
\end{table}

Figure~\ref{fig:synthetic_data} evaluates orders $p=0,\ldots,5$ using $s\in\{8,16,32,64,128,256\}$ independent Rademacher probes. As the number of probes increases, the matrix-free runtime grows approximately linearly, while the relative RMS trace error decreases consistently across all orders. The same convergence behavior is observed for the second spectral moment: its relative RMS error decreases from
approximately $1.9\%$ at $s=8$ to $0.33\%$ at $s=256$, while the trace error decreases from about $1.1\%$ to $0.19\%$. The synthetic benchmark confirms the expected accuracy-cost tradeoff of the
matrix-free estimator. In the binding-affinity experiments, however, predictive accuracy changes little over the tested probe counts, indicating that increased numerical precision does not necessarily translate into improved prediction in the present setting.

\begin{figure}[t]
\centering
\includegraphics[width=\linewidth]{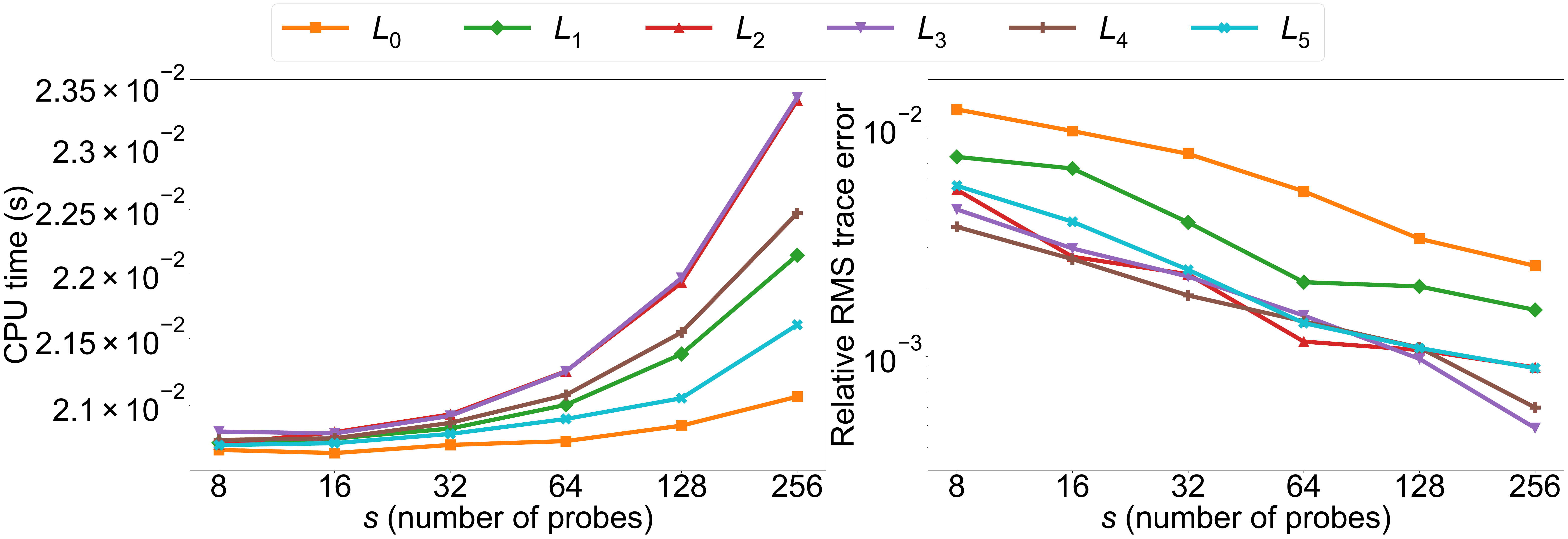}
\caption{
Synthetic matrix-free benchmark for orders $L_0$--$L_5$. Left: CPU time as a function of the number of probe vectors. Right: relative RMS trace error. Increasing the probe count increases the computational cost approximately linearly while reducing the stochastic trace-estimation error.
}
\label{fig:synthetic_data}
\end{figure}

\subsection{Topology and Protein-Language Feature Ablation}\label{subsec:why_works}

To examine how the matrix-free topological descriptors behave alone and together with sequence-derived information, we compare four feature settings using the same tuned Huber GBDT model and the same 10-fold cross-validation protocol: persistent hyperdigraph features, directed-bipartite features, persistent hyperdigraph features with ESM, and directed-bipartite features with ESM in Table \ref{tab:topology-esm-ablation}. For all features, we use \(s=16\) probe vectors.

The topology-only descriptors provide useful but moderate predictive signals. On WT, the directed-bipartite representation outperforms MFPHL, with $R_p=0.5326$ versus $0.4913$, whereas on V2020 the higher-order
MFPHL representation is stronger. Adding ESM substantially improves both topological representations across all three datasets, and the directed--bipartite–ESM combination gives the strongest performance among
the four settings considered here. Because an ESM-only model is not included in this ablation, these results characterize how the proposed topological descriptors behave when combined with ESM rather than isolating their incremental gain over sequence information alone. Published ESM-only results from related models are available in Table~5 of \cite{xu2025plnet} for reference.
\begin{table}[!htbp]
\centering
\small
\caption{Topology and ESM feature ablation for binding-affinity prediction,
reported as mean \(\pm\) standard deviation over 10 random seeds.}
\label{tab:topology-esm-ablation}
\begin{tabular}{lccc}
\hline
Method & WT & V2020 & P2P$_{\mathrm{WT}}$ \\
\hline
\multicolumn{4}{l}{\text{Pearson} $R_p$} \\
MFPHL \(L_0\)--\(L_5\)
& 0.4913\,$\pm$\,0.0148
& 0.3717\,$\pm$\,0.0049
& 0.4020\,$\pm$\,0.0046 \\
Directed-bipartite \(d_1\)
& 0.5326\,$\pm$\,0.0109
& 0.2851\,$\pm$\,0.0077
& 0.3401\,$\pm$\,0.0040 \\
MFPHL \(L_0\)--\(L_5\) + ESM
& 0.7004\,$\pm$\,0.0067
& 0.6075\,$\pm$\,0.0022
& 0.6285\,$\pm$\,0.0022 \\
Directed-bipartite \(d_1\) + ESM
& 0.7276\,$\pm$\,0.0078
& 0.6124\,$\pm$\,0.0026
& 0.6355\,$\pm$\,0.0015 \\
\hline
\multicolumn{4}{l}{\text{MAE}} \\
MFPHL \(L_0\)--\(L_5\)
& 0.1152\,$\pm$\,0.0016
& 0.1157\,$\pm$\,0.0005
& 0.1155\,$\pm$\,0.0004 \\
Directed-bipartite \(d_1\)
& 0.1135\,$\pm$\,0.0009
& 0.1201\,$\pm$\,0.0003
& 0.1195\,$\pm$\,0.0002 \\
MFPHL \(L_0\)--\(L_5\) + ESM
& 0.0899\,$\pm$\,0.0005
& 0.0965\,$\pm$\,0.0003
& 0.0953\,$\pm$\,0.0003 \\
Directed-bipartite \(d_1\) + ESM
& 0.0862\,$\pm$\,0.0009
& 0.0961\,$\pm$\,0.0002
& 0.0947\,$\pm$\,0.0002 \\
\hline
\end{tabular}
\end{table}
% \subsection{The Absence of Betti Numbers}\label{subsec:no_betti}

% We next examined whether explicit Betti-number information improves the proposed feature representation. For this test, direct Betti descriptors were computed by rank calculation for the element-specific persistent Laplacians over the same \(L_0,\ldots,L_5\) levels and cutoffs \(3,\ldots,12\), and then concatenated with the MFPHL hyperdigraph descriptor, the directed-bipartite descriptor, and ESM features. Using the same tuned Huber GBDT and the same 10-seed stratified 10-fold protocol, adding Betti numbers did not improve prediction performance. The resulting \(R_p\) values were \(0.7239\) on WT, \(0.6127\) on V2020, and \(0.6338\) on P2P\(_{\mathrm{WT}}\), with MAE values \(0.0877\), \(0.0961\), and \(0.0949\), respectively. These values are essentially unchanged and, in WT, slightly lower compared with the corresponding descriptors without Betti numbers\revised{: \(R_p\) changes by at most \(0.0004\) and MAE is unchanged (Table~\ref{tab:mfphl-cumulative-order-ablation}, \(L_0\)--\(L_5\)). One possible explanation, not tested by this experiment, is that the connectivity and cavity information in the Betti counts is already reflected indirectly in the other descriptors. The absence of a gain does not by itself establish this}.

\subsection{Single-order Hyperdigraph Contributions}
\label{single-order}

To examine the contribution of individual hyperdigraph orders, we evaluate descriptors based on $L_0,\ldots,L_5$, each combined with the same directed-bipartite and ESM features. As shown in
Table~\ref{tab:single-order}, the six variants achieve similar predictive performance across all three datasets. On WT, $R_p$ ranges from $0.7251$ for $L_0$ to $0.7293$ for $L_1$; on V2020, it ranges from
$0.6117$ to $0.6137$; and on P2P$_{\mathrm{WT}}$, from $0.6349$ to $0.6357$. These differences are comparable to or smaller than the variation across training seeds.

The consistent performance across orders indicates that useful structural information is not confined to a single hyperdigraph level. At the same time, the small differences among the single-order models do not provide a clear basis for selecting one order in preference to the others. Accordingly, in the full PHL and MFPHL representations, we retain $L_0$--$L_5$ together, allowing the model to make use of information
represented across multiple hyperdigraph orders rather than relying on a single prescribed order.

\begin{table}[!htbp]
\centering
\small
\setlength{\tabcolsep}{6pt}
\caption{Single-order ablation. Each setting uses one hyperdigraph Laplacian order together with the directed-bipartite and ESM descriptors, with \(s=16\) probe vectors and the same tuned Huber GBDT model. Values are reported as mean \(\pm\) standard deviation over ten random seeds.}
\label{tab:single-order}
\begin{tabular}{lccc}
\hline
Orders & WT & V2020 & P2P$_{\mathrm{WT}}$ \\
\hline
\multicolumn{4}{l}{\(R_p\)} \\
\(L_0\) & 0.7251\,$\pm$\,0.0040 & 0.6137\,$\pm$\,0.0015 & 0.6356\,$\pm$\,0.0015 \\
\(L_1\) & 0.7293\,$\pm$\,0.0067 & 0.6126\,$\pm$\,0.0023 & 0.6357\,$\pm$\,0.0018 \\
\(L_2\) & 0.7281\,$\pm$\,0.0060 & 0.6133\,$\pm$\,0.0021 & 0.6354\,$\pm$\,0.0015 \\
\(L_3\) & 0.7284\,$\pm$\,0.0058 & 0.6135\,$\pm$\,0.0019 & 0.6356\,$\pm$\,0.0017 \\
\(L_4\) & 0.7289\,$\pm$\,0.0056 & 0.6123\,$\pm$\,0.0026 & 0.6350\,$\pm$\,0.0014 \\
\(L_5\) & 0.7286\,$\pm$\,0.0052 & 0.6117\,$\pm$\,0.0027 & 0.6349\,$\pm$\,0.0013 \\
\hline
\multicolumn{4}{l}{MAE} \\
\(L_0\) & 0.0869\,$\pm$\,0.0004 & 0.0960\,$\pm$\,0.0002 & 0.0947\,$\pm$\,0.0002 \\
\(L_1\) & 0.0859\,$\pm$\,0.0009 & 0.0961\,$\pm$\,0.0003 & 0.0947\,$\pm$\,0.0002 \\
\(L_2\) & 0.0866\,$\pm$\,0.0008 & 0.0961\,$\pm$\,0.0002 & 0.0947\,$\pm$\,0.0002 \\
\(L_3\) & 0.0867\,$\pm$\,0.0008 & 0.0960\,$\pm$\,0.0002 & 0.0946\,$\pm$\,0.0002 \\
\(L_4\) & 0.0862\,$\pm$\,0.0007 & 0.0961\,$\pm$\,0.0002 & 0.0947\,$\pm$\,0.0002 \\
\(L_5\) & 0.0865\,$\pm$\,0.0006 & 0.0962\,$\pm$\,0.0003 & 0.0948\,$\pm$\,0.0002 \\
\hline
\end{tabular}
\end{table}

\section{Conclusion}\label{Conclusion}
This work introduced persistent hyperdigraph learning (PHL) for protein-protein binding affinity prediction and developed its matrix-free counterpart, MFPHL, to reduce feature-generation costs. PHL incorporates directed and many-body relationships into multiscale spectral representations of molecular interfaces. MFPHL replaces explicit Laplacian assembly and eigendecomposition with randomized probing through sparse boundary operators. Its cost per probe is governed by operator sparsity, with dimension-independent probe bounds for prescribed relative accuracy and confidence of the trace estimate. In our experiments, this formulation reduces feature-generation time by approximately two orders of magnitude at a small probe count, making the exploration of higher topological orders substantially more affordable.

Among the methods compared, PHL achieves the highest correlation on all three datasets, although the comparator values are published results rather than paired reruns. On WT, the gap between MFPHL and PHL is smaller than the training-seed variability. However, on V2020 and P2P\(_{\mathrm{WT}}\), the matrix-free descriptor exhibits a clear accuracy gap relative to PHL, although it exceeds the reported PLNet correlations. Increasing the probe count does not systematically improve predictions over the tested range. These results demonstrate a dataset-dependent trade-off between the computational efficiency of MFPHL and the predictive accuracy of the eigenvalue-based PHL descriptor.

The MFPHL ablation studies further clarify the role of the topological information. Adding persistent Betti numbers changes \(R_p\) by at most \(0.0004\), while first-order bipartite features combined with sequence embeddings achieve performance comparable to the full high-order descriptor across all datasets considered. Thus, making higher-order features computationally accessible does not necessarily translate into improved prediction. These findings apply to the tested MFPHL descriptors and do not establish that higher-order structure is uninformative in PHL more generally. Its contribution may be redundant with lower-order and sequence features or insufficiently resolved by the moment estimates and probe statistics used here.

Future work will explore richer matrix-free spectral summaries, including higher moments and spectral density estimates, to assess whether they narrow the accuracy gap between MFPHL and PHL and better distinguish contributions from higher topological orders. The same approach may also extend to other persistent operators, provided their factorizations support efficient matrix-vector products. Finally, reducing the need for truncation in large, dense channels would allow a more complete evaluation of higher-order representations and help separate the effects of computational approximations from those of descriptor design.

\section*{Data \& Code Availability}

The code needed to reproduce this paper's results is available on request
from the corresponding author. The SKEMPI v2 dataset is available at
\url{https://life.bsc.es/pid/skempi2/}, and the PDBbind V2020 dataset at
\url{http://www.pdbbind.org.cn/}.

\section*{Acknowledgments} J. C. was partially supported by the National Science Foundation under awards DMS/NIGMS 2553768, DMS 2514195, the Arkansas Biosciences Institute Gap Grant, and the computational resources provided by the Arkansas High Performance Computing Center (AHPCC). C. W. was partially supported   by the National Science Foundation under award DMS-2206332.

\bibliographystyle{plain}
\bibliography{main}

\end{document}

\maketitle
\section{Proof of Theorem~\ref{thm_hutchinson}}\label{proof_thm_hutchinson}

\begin{proof}[Proof of Theorem~\ref{thm_hutchinson}]
Write \(L=L_p^{a,b}\). Since the entries are independent
Rademacher variables, \(\mathbb{E}[z_i]=0\) and \(z_i^2=1\).
Thus \(\mathbb{E}[z_i z_j]=\delta_{ij}\), giving
\(\mathbb{E}[zz^\top]=I_n\). Consequently,
\[
  \mathbb{E}[z^\top Lz]
  =\operatorname{tr}\bigl(L\,\mathbb{E}[zz^\top]\bigr)
  =\operatorname{tr}(L).
\]

For the variance, symmetry of \(L\) and \(z_i^2=1\) give
\[
  z^\top Lz-\operatorname{tr}(L)
  =2\sum_{i<j}L_{ij}z_i z_j.
\]
For \(i<j\) and \(k<\ell\), independence and the zero means imply
\[
  \mathbb{E}[z_i z_j z_k z_\ell]
  =
  \begin{cases}
    1, & (i,j)=(k,\ell),\\
    0, & \text{otherwise}.
  \end{cases}
\]
It follows that
\[
  \operatorname{Var}(z^\top Lz)
  =4\sum_{i<j}L_{ij}^{2}
  =2\left(\|L\|_F^{2}-\sum_i L_{ii}^{2}\right)
  \le 2\|L\|_F^{2}.
\]
Since \(L\) is symmetric, \(\|L\|_F^2=\sum_i\lambda_i^2\).
Hence the first two spectral moments satisfy
\[
  \sum_i\lambda_i=\mathbb{E}[z^\top Lz],
  \qquad
  \sum_i\lambda_i^2
  =\frac12\operatorname{Var}(z^\top Lz)+\sum_iL_{ii}^2.
\]

For the matrix-free factorization, substitute
\[
  L={B_{p+1}^{a,b}(B_{p+1}^{a,b})^\top
    +(B_p^a)^\top B_p^a}
\]
to obtain
\[
  z^\top Lz
  ={\|(B_{p+1}^{a,b})^\top z\|_2^2
   +\|B_p^a z\|_2^2}.
\]
The diagonal correction is also available directly from these
operators, since
\[
  L_{ii}
  ={\sum_k(B_{p+1}^{a,b})_{ik}^{2}
   +\sum_k(B_p^a)_{ki}^{2}}.
\]
Computing all diagonal entries therefore requires
\(\mathcal{O}(N+n)\) operations without assembling \(L\).

Each quadratic-form evaluation uses two sparse matrix--vector
products. The products and squared norms can be evaluated in
\(\mathcal{O}(N+n)\) operations by storing only active output
coordinates, of which there are at most \(N\).
The boundary matrices require \(\mathcal{O}(N)\) storage in a
coordinate sparse representation, and the probes and working
vectors require \(\mathcal{O}(N+n)\) total storage.
No explicit Laplacian assembly is needed.
\end{proof}

\begin{remark}[Gaussian probes]
Gaussian probes \(g\sim\mathcal{N}(0,I_n)\) also satisfy \(\mathbb{E}[gg^\top]=I_n\), and hence \(\mathbb{E}[g^\top Lg]=\operatorname{tr}(L)\).
Writing \(L=U\Lambda U^\top\), orthogonal invariance gives \(U^\top g\sim\mathcal{N}(0,I_n)\). Independence of its entries
and \(\operatorname{Var}(g_i^2)=2\) yields
\(
  \operatorname{Var}(g^\top Lg)
  =2\sum_i\lambda_i^2
  =2\|L\|_F^2.
\)
Thus, Rademacher probes have strictly smaller variance whenever \(L\) has a nonzero diagonal entry, while the variance bound \(2\|L\|_F^2\) applies to both distributions.
\end{remark}

\section{Proof of Proposition~\ref{prop_annihilate}}
\label{proof_prop_annihilate}

\begin{proof}
Both summands of \eqref{eq_quadratic_form} are squared Euclidean norms
and hence non-negative, so $z^{\top}L_{p}^{a,b}z=0$ if and only if
\begin{equation*}
  {\big(B_{p+1}^{a,b}\big)^{\top}z=0
  \qquad\text{and}\qquad
  B_{p}^{a}z=0} .
\end{equation*}
In operator terms, these say {$\big(d_{p+1}^{a,b}\big)^{\ast}z=0$ and
$d_{p}^{a}z=0$}, that is
$z\in\ker d_{p}^{a}\cap\ker\big(d_{p+1}^{a,b}\big)^{\ast}$, which by
Theorem~\ref{thm_persistent_hodge} is exactly $\ker\Delta_{p}^{a,b}$.

For the second claim, write $z=z_{\mathrm{harm}}+z_{\perp}$ with
$z_{\mathrm{harm}}\in\ker\Delta_{p}^{a,b}$ and
$z_{\perp}\in\Ima d_{p+1}^{a,b}\oplus\Ima\big(d_{p}^{a}\big)^{\ast}$,
the decomposition being orthogonal by
Theorem~\ref{thm_persistent_hodge}. Since
$\Delta_{p}^{a,b}z_{\mathrm{harm}}=0$,
\begin{equation*}
  z^{\top}L_{p}^{a,b}z
  =\big(z_{\mathrm{harm}}+z_{\perp}\big)^{\top}
   L_{p}^{a,b}\big(z_{\mathrm{harm}}+z_{\perp}\big)
  =z_{\perp}^{\top}L_{p}^{a,b}z_{\perp},
\end{equation*}
the cross terms vanishing because
$L_{p}^{a,b}z_{\mathrm{harm}}=0$ and $L_{p}^{a,b}$ is self-adjoint.
\end{proof}

\section{Proof of Corollary~\ref{cor_invisible}}
\label{proof_cor_invisible}

\begin{proof}
Let $0=\lambda_{1}=\cdots=\lambda_{\beta}<\lambda_{\beta+1}\le\cdots
\le\lambda_{n}$ be the eigenvalues of $L_{p}^{a,b}$, with
$\beta=\beta_{p}^{a,b}$ by Theorem~\ref{thm_persistent_hodge}. For a
function $f$ defined on the spectrum, the spectral mapping theorem gives
\begin{equation*}
  \operatorname{tr}\big(f(L_{p}^{a,b})\big)
  =\sum_{i=1}^{n}f(\lambda_{i})
  =\beta f(0)+\sum_{i>\beta}f(\lambda_{i}),
\end{equation*}
so if $f(0)=0$, the first term vanishes, and the trace depends only on the
multiset $\{\lambda_{i}:\lambda_{i}>0\}$ instead of $\beta_{p}^{a,b}$.
\end{proof}

\section{Proof of Theorem~\ref{thm_recovery}}\label{proof_thm_recovery}

\begin{proof}[Proof of Theorem~\ref{thm_recovery}]
Unbiasedness follows from Theorem~\ref{thm_hutchinson} and linearity of
expectation, since each summand has mean $\operatorname{tr}(L)$. The
summands are independent and identically distributed. For the Rademacher
or standard Gaussian probes {assumed in the theorem}, the variance bound {of Theorem~\ref{thm_hutchinson} and the Gaussian remark} gives
\begin{equation*}
  \operatorname{Var}(T_{s})
  = \frac{1}{s}\operatorname{Var}\big(z^{\top}Lz\big)
  \le \frac{2\lVert L\rVert_{F}^{2}}{s}.
\end{equation*}

Because $L$ is positive semidefinite its eigenvalues are non-negative,
so every cross term in the expansion of
$\big(\sum_{i}\lambda_{i}\big)^{2}$ is non-negative and
\begin{equation*}
  \lVert L\rVert_{F}^{2}
  = \sum_{i}\lambda_{i}^{2}
  \;\leq\; \Big(\sum_{i}\lambda_{i}\Big)^{2}
  = \operatorname{tr}(L)^{2}.
\end{equation*}

{Since $L$ is positive semidefinite and $L\ne0$ by hypothesis, $\operatorname{tr}(L)>0$. (If $L=0$, every probe value is zero and $T_{s}=\operatorname{tr}(L)$ exactly, but the event in \eqref{eq_probe_count} then has probability one, which is why this case is excluded.)} Chebyshev's inequality applied to $T_{s}$ with
deviation $\varepsilon\operatorname{tr}(L)$ gives
\begin{equation*}
  \mathbb{P}\Big(\big|T_{s} - \operatorname{tr}(L)\big|
  \geq \varepsilon\operatorname{tr}(L)\Big)
  \;\leq\; \frac{2\lVert L\rVert_{F}^{2}}
                {s\,\varepsilon^{2}\operatorname{tr}(L)^{2}}
  \;\leq\; \frac{2}{s\,\varepsilon^{2}},
\end{equation*}
which is at most $\delta$ precisely when
$s \geq 2/(\delta\varepsilon^{2})$. The requirement involves only
$\varepsilon$ and $\delta$, and no dimension.
\end{proof}

\begin{remark}
The final inequality is tight only when $L$ has rank one. Writing
$r(L) = \operatorname{tr}(L)^{2}/\lVert L\rVert_{F}^{2} \in [1,n]$, the
bound reads $2/(s\varepsilon^{2}r(L))$, so
$s \geq 2/(\delta\varepsilon^{2}r(L))$ already suffices; the stated form
is uniform over all operators and pessimistic for those arising here.
Sharper bounds, with $\log(1/\delta)$ in place of $1/\delta$, follow from
concentration inequalities \cite{avron2011randomized,roosta2015improved}.
\end{remark}

\section{Proof of Theorem~\ref{thm_sandwich}}
\label{proof_thm_sandwich}

\begin{proof}
Write $L=L_{p}^{a,b}$, $B=B_{p+1}^{a,b}$ and $C=B_{p}^{a}$, so that by
Definition~\ref{def_persistent_laplacian}
\begin{equation}\label{eq_sm_L}
  L={BB^{\top}+C^{\top}C},
\end{equation}
a symmetric positive semidefinite matrix of order
$n=\dim\Omega_{p}^{a}$ with eigenvalues
$0\le\lambda_{1}\le\cdots\le\lambda_{n}$, of which exactly
$\beta:=\beta_{p}^{a,b}$ vanish by
Theorem~\ref{thm_persistent_hodge}. Recall
\begin{equation}\label{eq_sm_sr}
  r(L)=\frac{\operatorname{tr}(L)^{2}}{\lVert L\rVert_{F}^{2}}
      =\frac{\big(\sum_{i}\lambda_{i}\big)^{2}}{\sum_{i}\lambda_{i}^{2}},
  \qquad
  \sigma(M)=\sqrt{\lVert M\rVert_{1}\lVert M\rVert_{\infty}},
  \qquad
  \Lambda_{p}=\sigma(B)^{2}+\sigma(C)^{2},
\end{equation}
with $\lVert M\rVert_{1}=\max_{j}\sum_{i}|M_{ij}|$ and
$\lVert M\rVert_{\infty}=\max_{i}\sum_{j}|M_{ij}|$. Since $L\neq0$, both
$\operatorname{tr}(L)$ and $\lambda_{n}$ are positive.

For the rightmost inequality of \eqref{eq_sandwich}, observe that
$\lambda_{i}=0$ for $i\le\beta$, so both sums in \eqref{eq_sm_sr} range
effectively over the $n-\beta$ indices with $\lambda_{i}>0$. The
Cauchy--Schwarz inequality applied to the vectors $(1,\dots,1)$ and
$(\lambda_{i})_{\lambda_{i}>0}$, each of length $n-\beta$, gives
\begin{equation*}
  \Big(\sum_{\lambda_{i}>0}\lambda_{i}\Big)^{2}
  \;\le\;(n-\beta)\sum_{\lambda_{i}>0}\lambda_{i}^{2},
\end{equation*}
and dividing by $\sum_{i}\lambda_{i}^{2}>0$ yields $r(L)\le n-\beta$.
Equality in Cauchy--Schwarz holds exactly when the two vectors are
parallel, that is when all non-zero eigenvalues coincide, which is the
stated equality case.

For the middle inequality, $0\le\lambda_{i}\le\lambda_{n}$ for every $i$
gives
\begin{equation*}
  \sum_{i}\lambda_{i}^{2}=\sum_{i}\lambda_{i}\cdot\lambda_{i}
  \;\le\;\lambda_{n}\sum_{i}\lambda_{i}
  =\lambda_{n}\operatorname{tr}(L),
\end{equation*}
so that, substituting into \eqref{eq_sm_sr},
\begin{equation}\label{eq_sm_middle}
  r(L)\;\ge\;\frac{\operatorname{tr}(L)^{2}}
                  {\lambda_{n}\operatorname{tr}(L)}
  =\frac{\operatorname{tr}(L)}{\lambda_{n}} .
\end{equation}

The leftmost inequality requires bounding the numerator and denominator
of \eqref{eq_sm_middle} separately. For the numerator, any real matrix
$M$ satisfies
$\operatorname{tr}(M^{\top}M)=\operatorname{tr}(MM^{\top})
=\sum_{i,j}M_{ij}^{2}=\lVert M\rVert_{F}^{2}$, so applying this to each
summand of \eqref{eq_sm_L} and using linearity of the trace,
\begin{equation}\label{eq_sm_trace}
  \operatorname{tr}(L)=\lVert B\rVert_{F}^{2}+\lVert C\rVert_{F}^{2}.
\end{equation}
For the denominator, note first that any real $M$ and any vector $x$
satisfy, on writing
$|M_{ij}x_{j}|=|M_{ij}|^{1/2}\cdot|M_{ij}|^{1/2}|x_{j}|$ and applying
the Cauchy-Schwarz inequality to the inner sum,
\begin{align*}
  \lVert Mx\rVert_{2}^{2}
  &=\sum_{i}\Big(\sum_{j}M_{ij}x_{j}\Big)^{2}
   \;\le\;\sum_{i}\Big(\sum_{j}|M_{ij}|\Big)
          \Big(\sum_{j}|M_{ij}|\,x_{j}^{2}\Big)\\[2pt]
  &\le\;\lVert M\rVert_{\infty}\sum_{j}x_{j}^{2}\sum_{i}|M_{ij}|
   \;\le\;\lVert M\rVert_{\infty}\lVert M\rVert_{1}
          \lVert x\rVert_{2}^{2},
\end{align*}
whence $\lVert M\rVert_{2}\le\sigma(M)$. Since $L$ is symmetric positive
semidefinite, $\lambda_{n}=\lVert L\rVert_{2}$, and by \eqref{eq_sm_L},
the triangle inequality for the spectral norm and
$\lVert M^{\top}M\rVert_{2}=\lVert M\rVert_{2}^{2}$,
\begin{equation}\label{eq_sm_lammax}
  \lambda_{n}\;\le\;\lVert B\rVert_{2}^{2}+\lVert C\rVert_{2}^{2}
  \;\le\;\sigma(B)^{2}+\sigma(C)^{2}=\Lambda_{p}.
\end{equation}
Combining \eqref{eq_sm_middle}, \eqref{eq_sm_trace} and
\eqref{eq_sm_lammax},
\begin{equation*}
  r(L)\;\ge\;\frac{\operatorname{tr}(L)}{\lambda_{n}}
  \;\ge\;\frac{\lVert B\rVert_{F}^{2}+\lVert C\rVert_{F}^{2}}
              {\Lambda_{p}},
\end{equation*}
which is the leftmost inequality of \eqref{eq_sandwich}.

Finally, each of $\lVert B\rVert_{F}^{2}$, $\lVert C\rVert_{F}^{2}$,
$\lVert B\rVert_{1}$, $\lVert B\rVert_{\infty}$, $\lVert C\rVert_{1}$
and $\lVert C\rVert_{\infty}$ is a sum or a maximum over the stored
non-zero entries, computable in a single pass at cost
$\mathcal{O}(\operatorname{nnz}(B)+\operatorname{nnz}(C))
=\mathcal{O}(N)$, with no matrix product and no eigenvalue.
\end{proof}

\section{Proof of Corollary~\ref{cor_apriori}}
\label{proof_cor_apriori}

\begin{proof}
Retain the notation of Section~\ref{proof_thm_sandwich} and let
$z_{1},\dots,z_{s}$ be independent {probe vectors with independent Rademacher entries or with independent standard Gaussian entries, as in Theorem~\ref{thm_recovery}}, and
$T_{s}=s^{-1}\sum_{j=1}^{s}z_{j}^{\top}Lz_{j}$. By
Theorem~\ref{thm_hutchinson}, for Rademacher probes each summand has mean
$\operatorname{tr}(L)$ and variance at most $2\lVert L\rVert_{F}^{2}$.
The same bound holds for standard Gaussian probes. By independence,
\begin{equation}\label{eq_sm_var}
  \mathbb{E}[T_{s}]=\operatorname{tr}(L),
  \qquad
  \operatorname{Var}(T_{s})\le\frac{2\lVert L\rVert_{F}^{2}}{s}.
\end{equation}
{Since $L\ne0$ is positive semidefinite, $\operatorname{tr}(L)>0$.} Chebyshev's inequality
with deviation $\varepsilon\operatorname{tr}(L)$, followed by
\eqref{eq_sm_var} and the definition of $r(L)$ in \eqref{eq_sm_sr},
gives
\begin{equation}\label{eq_sm_cheb}
  \mathbb{P}\Big(\big|T_{s}-\operatorname{tr}(L)\big|
  \ge\varepsilon\operatorname{tr}(L)\Big)
  \;\le\;\frac{\operatorname{Var}(T_{s})}
              {\varepsilon^{2}\operatorname{tr}(L)^{2}}
  \;\le\;\frac{2\lVert L\rVert_{F}^{2}}
        {s\,\varepsilon^{2}\operatorname{tr}(L)^{2}}
  =\frac{2}{s\,\varepsilon^{2}\,r(L)} .
\end{equation}
Bounding $r(L)$ below by Theorem~\ref{thm_sandwich} and substituting
into \eqref{eq_sm_cheb},
\begin{equation*}
  \mathbb{P}\Big(\big|T_{s}-\operatorname{tr}(L)\big|
  \ge\varepsilon\operatorname{tr}(L)\Big)
  \;\le\;\frac{2\,\Lambda_{p}}
  {s\,\varepsilon^{2}\big(\lVert B\rVert_{F}^{2}
  +\lVert C\rVert_{F}^{2}\big)},
\end{equation*}
which is at most $\delta$ precisely when $s$ satisfies
\eqref{eq_apriori}. Every quantity appearing there is obtained from the
stored entries of $B$ and $C$ by the cost argument above, so the bound is
available before any spectral computation.
\end{proof}

\section{Applicability Beyond Hyperdigraphs}
\label{subsec:graph_laplacian}

The estimator of Theorem~\ref{thm_hutchinson} extends naturally beyond persistent hyperdigraph Laplacians. Its key requirement is a sparse factorization of the form \(B^{\top}B\), a structure shared by a broad class of discrete Laplacian operators. A canonical example is the ordinary graph Laplacian, for which \(B\) is the signed incidence matrix. Since the order-zero operator in our construction reduces precisely to this case, graph Laplacians provide a simple setting in which to isolate the computational scaling of the estimator and verify that its advantages are not tied to the hyperdigraph construction.

We therefore generate \(n\) points uniformly in \([0,1]^3\) and connect pairs whose Euclidean distance is at most \(\epsilon\), yielding the graph Laplacian \(L_0(\epsilon)\). The probe estimator uses (s) probe vectors and the eight statistics of Section~\ref{hyperdigraph_construction}, while the spectral baseline computes the dense eigendecomposition and forms the corresponding eight eigenvalue statistics. All timings are averaged over ten random seeds.

Table~\ref{tab:synthetic-baseline} varies the matrix size at
$s = 16$. The probe time tracks the number of non-zero entries rather
than the dimension, rising from $69$ to $264$ microseconds as
$\operatorname{nnz}(L_{0})$ grows by a factor of ninety, while the
baseline grows by two orders of magnitude over the same range, so the
speedup increases from roughly $2\times$ at $n = 64$ to $149\times$ at
$n = 768$. Table~\ref{tab:synthetic-sweep} instead fixes $n = 768$ and
$\epsilon = 0.20$ and varies $s$. Probe time grows linearly, as each
probe requires {one sparse product with the incidence matrix}, and the advantage
narrows accordingly; the estimator remains faster than dense
eigendecomposition by more than an order of magnitude even at
$s = 256$.

\begin{table}[!htbp]
\centering
\small
\caption{Runtime on synthetic graph Laplacians at $s = 16$ probe vectors, averaged over ten seeds.}
\label{tab:synthetic-baseline}
\begin{tabular}{rrrrrrr}
\hline
$n$ & $\epsilon$ & Edges & nnz$(L_{0})$
    & Probe (ms) & Baseline (ms) & Speedup \\
\hline
  64 & 0.16 &    28 &    120 & 0.065 & 0.180 & 2.8$\times$ \\
  64 & 0.20 &    57 &    178 & 0.060 & 0.165 & 2.8$\times$ \\
  64 & 0.24 &    94 &    252 & 0.057 & 0.171 & 3.0$\times$ \\
 128 & 0.16 &   112 &    352 & 0.064 & 0.526 & 8.3$\times$ \\
 128 & 0.20 &   211 &    550 & 0.063 & 0.582 & 9.2$\times$ \\
 128 & 0.24 &   351 &    830 & 0.065 & 0.615 & 9.4$\times$ \\
 256 & 0.16 &   460 &   1176 & 0.081 & 3.827 & 47.4$\times$ \\
 256 & 0.20 &   862 &   1980 & 0.084 & 3.693 & 44.1$\times$ \\
 256 & 0.24 &  1420 &   3096 & 0.086 & 3.754 & 43.6$\times$ \\
 512 & 0.16 &  1879 &   4270 & 0.142 & 18.170 & 128.2$\times$ \\
 512 & 0.20 &  3487 &   7486 & 0.154 & 18.237 & 118.8$\times$ \\
 512 & 0.24 &  5712 &  11937 & 0.164 & 18.334 & 111.8$\times$ \\
 768 & 0.16 &  4198 &   9165 & 0.184 & 38.601 & 210.1$\times$ \\
 768 & 0.20 &  7827 &  16421 & 0.218 & 38.779 & 178.3$\times$ \\
 768 & 0.24 & 12840 &  26448 & 0.275 & 38.790 & 141.0$\times$ \\
\hline
\end{tabular}
\end{table}

\begin{table}[!htbp]
\centering
\small
\caption{Runtime as the number of probe vectors increases, at $n = 768$ and $\epsilon = 0.20$.}
\label{tab:synthetic-sweep}
\begin{tabular}{rrrr}
\hline
$s$ & Probe (ms) & Baseline (ms) & Speedup \\
\hline
  16 & 0.231 & 38.931 & 168.9$\times$ \\
  32 & 0.377 & 38.948 & 103.4$\times$ \\
  64 & 0.606 & 38.918 & 64.2$\times$ \\
 128 & 1.192 & 39.063 & 32.8$\times$ \\
 256 & 2.297 & 39.138 & 17.0$\times$ \\
\hline
\end{tabular}
\end{table}

\section{Effect of Higher-Order Hyperdigraph Terms}
\label{subsec:order_limit}

To assess the contribution of higher-order persistent hyperdigraph information, we perform a cumulative-order ablation using the MFPHL descriptor together with the directed-bipartite feature and ESM embedding under \(s=16\). Table~\ref{tab:mfphl-cumulative-order-ablation} shows that simply increasing the hyperdigraph order does not lead to monotonic improvement. For WT, the best overall performance is obtained with \(L_0\)-\(L_1\), reaching \(R_p=0.7259\) and MAE \(=0.0869\), compared with \(R_p=0.7251\) for \(L_0\) alone. After adding \(L_2\) and \(L_3\), however, the correlation decreases to \(0.7247\) and \(0.7226\), respectively. Interestingly, the performance then gradually recovers as additional orders are incorporated, reaching \(R_p=0.7243\) for \(L_0\)-\(L_5\). {Thus \(L_0\)--\(L_5\) lies above \(L_0\)--\(L_3\) and \(L_0\)--\(L_4\) but slightly below \(L_0\)--\(L_2\) and the strongest low-order model. All of these WT differences, at most \(0.0043\), are smaller than the corresponding standard deviations of \(0.004\) to \(0.008\).} In contrast, for V2020 and P2P$_{\mathrm{WT}}$, \(L_0\) alone gives the highest correlation, with \(R_p=0.6137\) and \(0.6356\), respectively, and the addition of \(L_1\)-\(L_4\) produces only small fluctuations or slight degradation.

{Because these changes are within training-seed variability, we do not attribute them to specific structural mechanisms. They describe how predictive performance behaves as higher-order descriptors become computationally accessible. Under this model and feature set, adding orders beyond \(L_0\) or \(L_0\)--\(L_1\) does not improve prediction, and the WT correlation declines slightly from \(R_p=0.7243\) for \(L_0\)--\(L_5\) to \(0.7216\) for \(L_0\)--\(L_8\). Possible explanations, not tested here, include redundancy between orders, the increasing sparsity and geometric sensitivity of very high-order hyperedges, and the larger feature dimension.}
\begin{table}[!htbp]
\centering
\small
\setlength{\tabcolsep}{6pt}
\caption{Cumulative-order ablation. All settings include the directed-bipartite and ESM descriptors, use \(s=16\) probe vectors and the same tuned Huber GBDT; values are mean \(\pm\) standard deviation
over ten random seeds.}
\label{tab:mfphl-cumulative-order-ablation}
\begin{tabular}{lccc}
\hline
Orders & WT & V2020 & P2P$_{\mathrm{WT}}$ \\
\hline
\multicolumn{4}{l}{\(R_p\)} \\
$L_0$        & 0.7251\,$\pm$\,0.0040 & 0.6137\,$\pm$\,0.0015 & 0.6356\,$\pm$\,0.0015 \\
$L_0$--$L_1$ & 0.7259\,$\pm$\,0.0053 & 0.6130\,$\pm$\,0.0018 & 0.6342\,$\pm$\,0.0025 \\
$L_0$--$L_2$ & 0.7247\,$\pm$\,0.0074 & 0.6115\,$\pm$\,0.0018 & 0.6341\,$\pm$\,0.0024 \\
$L_0$--$L_3$ & 0.7226\,$\pm$\,0.0063 & 0.6119\,$\pm$\,0.0032 & 0.6344\,$\pm$\,0.0024 \\
$L_0$--$L_4$ & 0.7237\,$\pm$\,0.0058 & 0.6127\,$\pm$\,0.0028 & 0.6341\,$\pm$\,0.0019 \\
$L_0$--$L_5$ & 0.7243\,$\pm$\,0.0061 & 0.6127\,$\pm$\,0.0023 & 0.6339\,$\pm$\,0.0018\\
$L_0$--$L_6$ & 0.7227\,$\pm$\,0.0080 & 0.6118\,\(\pm\)0.0020   & 0.6342\,\(\pm\)0.0023\\
$L_0$--$L_7$ & 0.7217\,$\pm$\,0.0055 & 
0.6099\,\(\pm\)0.0026 & 0.6338\,$\pm$\,0.0017\\
$L_0$--$L_8$ & 0.7216\,$\pm$\,0.0056 &        0.6099\,\(\pm\)0.0027 & 0.6337\,$\pm$\,0.0017\\
\hline
\multicolumn{4}{l}{MAE} \\
$L_0$        & 0.0869\,$\pm$\,0.0004 & 0.0960\,$\pm$\,0.0002 & 0.0947\,$\pm$\,0.0002 \\
$L_0$--$L_1$ & 0.0869\,$\pm$\,0.0006 & 0.0961\,$\pm$\,0.0002 & 0.0948\,$\pm$\,0.0004 \\
$L_0$--$L_2$ & 0.0872\,$\pm$\,0.0009 & 0.0962\,$\pm$\,0.0002 & 0.0949\,$\pm$\,0.0003 \\
$L_0$--$L_3$ & 0.0878\,$\pm$\,0.0007 & 0.0961\,$\pm$\,0.0003 & 0.0948\,$\pm$\,0.0003 \\
$L_0$--$L_4$ & 0.0877\,$\pm$\,0.0007 & 0.0961\,$\pm$\,0.0003 & 0.0948\,$\pm$\,0.0003 \\
$L_0$--$L_5$ & 0.0877\,$\pm$\,0.0006 & 0.0961\,$\pm$\,0.0003 & 0.0949\,$\pm$\,0.0003                  \\
$L_0$--$L_6$ & 0.0878\,$\pm$\,0.0009 & 0.0962\,\(\pm\)0.0002  & 0.0949\,\(\pm\)0.0003\\
$L_0$--$L_7$ & 0.0881\,$\pm$\,0.0006 &         0.0964\,\(\pm\)0.0003     & 0.0950\,\(\pm\)0.0003\\
$L_0$--$L_8$ & 0.0881\,$\pm$\,0.0007 &          0.0964\,\(\pm\)0.0003   & 0.0950\,\(\pm\)0.0003\\
\hline
\end{tabular}
\end{table}

\bibliographystyle{plain}
\bibliography{main}